\ifx\documentclass\undefined
\documentstyle[12pt]{article}
\else
\documentclass[12pt]{article}
\fi

\usepackage{pst-all}
\usepackage{float}
\usepackage{amssymb}
\usepackage{graphicx}
\usepackage{color}

\newcommand{\beq}{\begin{eqnarray*}}
\newcommand{\eeq}{\end{eqnarray*}}
\makeatletter
\renewcommand{\theequation}{\thesection.\arabic{equation}}
\@addtoreset{equation}{section}
\def\eqnarray{%
\stepcounter{equation}%
\let\@currentlabel=\theequation
\global\@eqnswtrue
\global\@eqcnt\z@
\tabskip\@centering
\let\\=\@eqncr
$$\halign to \displaywidth\bgroup\@eqnsel\hskip\@centering
$\displaystyle\tabskip\z@{##}$&\global\@eqcnt\@ne
\hfil$\displaystyle{{}##{}}$\hfil
&\global\@eqcnt\tw@$\displaystyle\tabskip\z@{##}$\hfil
\tabskip\@centering&\llap{##}\tabskip\z@\cr}
\makeatother

\def\bbbz{{\mathchoice {\hbox{$\sf\textstyle Z\kern-0.4em Z$}}
{\hbox{$\sf\textstyle Z\kern-0.4em Z$}}
{\hbox{$\sf\scriptstyle Z\kern-0.3em Z$}}
{\hbox{$\sf\scriptscriptstyle Z\kern-0.2em Z$}}}}
\def\bbbq{{\mathchoice {\setbox0=\hbox{$\displaystyle\rm Q$}\hbox{\raise
0.15\ht0\hbox to0pt{\kern0.4\wd0\vrule height0.8\ht0\hss}\box0}}
{\setbox0=\hbox{$\textstyle\rm Q$}\hbox{\raise
0.15\ht0\hbox to0pt{\kern0.4\wd0\vrule height0.8\ht0\hss}\box0}}
{\setbox0=\hbox{$\scriptstyle\rm Q$}\hbox{\raise
0.15\ht0\hbox to0pt{\kern0.4\wd0\vrule height0.7\ht0\hss}\box0}}
{\setbox0=\hbox{$\scriptscriptstyle\rm Q$}\hbox{\raise
0.15\ht0\hbox to0pt{\kern0.4\wd0\vrule height0.7\ht0\hss}\box0}}}}
\def\bbbc{{\mathchoice {\setbox0=\hbox{$\displaystyle \rm C$}\hbox{\raise
0.06\ht0\hbox to0pt{\kern0.4\wd0\vrule height0.9\ht0\hss}\box0}}
{\setbox0=\hbox{$\textstyle\rm C$}\hbox{\raise
0.06\ht0\hbox to0pt{\kern0.4\wd0\vrule height0.9\ht0\hss}\box0}}
{\setbox0=\hbox{$\scriptstyle\rm C$}\hbox{\raise
0.06\ht0\hbox to0pt{\kern0.4\wd0\vrule height0.8\ht0\hss}\box0}}
{\setbox0=\hbox{$\scriptscriptstyle\rm C$}\hbox{\raise
0.06\ht0\hbox to0pt{\kern0.4\wd0\vrule height0.8\ht0\hss}\box0}}}}
\makeatletter
  \renewcommand{\theequation}{%
 \thesection.\arabic{equation}}
  \@addtoreset{equation}{section}
\makeatother
\newtheorem{theorem}{Theorem}[section]
\newtheorem{lemma}[theorem]{Lemma}
\newtheorem{corollary}[theorem]{Corollary}
\newtheorem{proposition}[theorem]{Proposition}
\newtheorem{remark}{Remark}[section]
\newtheorem{definition}{Definition}[section]
\newsavebox{\toy}
\savebox{\toy}{\framebox[0.65em]{\rule{0cm}{1ex}}}
\newcommand{\QED}{\usebox{\toy}}
\def\nlni{\par\ifvmode\removelastskip\fi\vskip\baselineskip\noindent}
\newenvironment{proof}{\nlni\begingroup\it Proof.\rm}{
\endgroup\vskip\baselineskip}

\begin{document}

\title{
Level statistics of one-dimensional 
Schr\"odinger operators with random decaying potential
}
\author{
Shinichi Kotani
\thanks{
Kwanseigakuin University, 
Sanda 669-1337, Japan.
e-mail : kotani@kwansei.ac.jp}
\and
Fumihiko Nakano
\thanks{
Department of Mathematics,
Gakushuin University,
1-5-1, Mejiro, Toshima-ku, Tokyo, 171-8588, Japan.
e-mail : 
fumihiko@math.gakushuin.ac.jp}
}
\maketitle
\begin{abstract}
We study 
the level statistics of one-dimensional Schr\"odinger operator 
with random potential decaying like 
$x^{-\alpha}$ 
at infinity.
We 
consider the point process 
$\xi_L$
consisting of the rescaled eigenvalues and show that : 
(i)(ac spectrum case)
for $\alpha > \frac 12$, 
$\xi_L$
converges to a clock process, and the fluctuation of the eigenvalue spacing converges to Gaussian.
(ii)(critical case)
for $\alpha = \frac 12$, 
$\xi_L$
converges to the limit of the circular $\beta$-ensemble. 
\end{abstract}
\section{Introduction}
\subsection{Background}
In this paper, 
we study the following Schr\"odinger operator
\beq
H &:=& - \frac {d^2}{dt^2} + a(t) F(X_t)
\quad
\mbox{on } 
L^2({\bf R})
\eeq
where
$a \in C^{\infty}$
is real valued, 
$a(-t) = a(t)$, 
non-increasing for 
$t \ge 0$, 
and satisfies
\[
C_1 t^{-\alpha} \le a(t) \le C_2 t^{-\alpha}, 
\quad
t \ge 1
\]
for some positive constants 
$C_1, C_2$
and 
$\alpha > 0$. 
$F$
is a real-valued, smooth, and non-constant function on a compact Riemannian manifold
$M$
such that
\[
\langle F \rangle := 
\int_M F(x) dx = 0.
\]
$\{ X_t \}$
is a Brownian motion on 
$M$.
Since 
the potential
$a(t) F(X_t)$
is 
$-\frac {d^2}{dt^2}$-compact, we have 
$\sigma_{ess}(H) = [0, \infty)$.
Kotani-Ushiroya\cite{KU} proved that
the spectrum of $H$
in 
$[0, \infty)$
is : 

(i)
for
$\alpha < \frac 12$ : 
pure point with exponentially decaying eigenfunctions, 

(ii)
for
$\alpha = \frac 12$ : 
pure point on
$[0, E_c]$
and purely singular continuous on 
$[E_c, \infty)$
with some explicitly computable
$E_c$, and

(iii)
for
$\alpha > \frac 12$ : 
purely absolutely continuous.

In this paper 
we study the level statistics of this operator.
For that purpose, let 
$H_L := H |_{[0, L]}$
be the local Hamiltonian with Dirichlet boundary condition and let 
$\{ E_n (L) \}_{n=1}^{\infty}$
be its eigenvalues in increasing order.
Let 
$n(L) \in {\bf N}$
be s.t.
$\{ E_n(L) \}_{n \ge n(L)}$
coincides with the set of positive eigenvalues of 
$H_L$. 
We take the reference energy 
$E_0 > 0$
arbitrarily and consider the point process
\[
\xi_L := 
\sum_{n \ge n(L)} 
\delta_{ L ( \sqrt{E_n(L)} - \sqrt{E_0} ) }
\]
in order to study the local fluctuation of eigenvalues near 
$E_0$. 
Our aim is to identify the limit of 
$\xi_L$
as 
$L \to \infty$.
Here 
we consider the scaling of 
$\sqrt{E_n(L)}$'s
instead of 
$E_n(L)$'s, 
which corresponds to the unfolding with respect to the density of states.
This problem 
was first studied by Molchanov\cite{Mol}.
He proved that, when 
$a(t)$
is constant, 
$\xi_L$
converges to the Poisson process. 
It was extended to the multidimensional Anderson model by Minami
\cite{Minami}.
Killip-Stoiciu \cite{KS} 
studied the CMV matrices whose matrix elements decay like 
$n^{-\alpha}$. 
They showed that $\xi_L$ converges to 

(i) $\alpha > \frac 12$ : 
the clock process, 

(ii) $\alpha = \frac 12$ : 
the limit of the circular $\beta$-ensemble, 

(iii) $0 < \alpha < \frac 12$ : 
the Poisson process.\\
Krichevski-Valk\'o-Vir\'ag\cite{KVV} studied the one-dimensional discrete Schr\"odinger operator 
with the random potential decaying like 
$n^{-1/2}$, 
and proved that 
$\xi_L$
converges to the $\mbox{Sine}_{\beta}$-process, which is the limit of the Gaussian $\beta$-ensemble found by Valk\'o-Vir\'ag\cite{VV}. 
The aim of our work 
is to do the analogue of that by Killip-Stoiciu\cite{KS} for the one-dimensional Schr\"odinger operator in the continuum. 
In subsection 1.2
(resp. subsection 1.3), 
we state our results for ac-case : $\alpha > \frac 12$
(resp. critical-case : $\alpha = \frac 12$).
We have not obtained results for pp-case : $\alpha < \frac 12$.

%
\subsection{AC-case}
\begin{definition}
Let 
$\mu$ 
be a probability measure on 
$[0, \pi)$.
We say that 
$\xi$
is the clock process with spacing 
$\pi$ 
with respect to  
$\mu$
if and only if 
\[
{\bf E}[ e^{- \xi(f)} ] 
=
\int_0^{\pi} d \mu(\phi)
\exp \left(
- \sum_{n \in {\bf Z}}
f(n \pi - \phi)
\right)
\]
where
$f \in C_c ({\bf R})$
and 
$\xi(f) := \int_{\bf R} f d \xi$. 
\end{definition}
We set 
\[
[x]_{\pi{\bf Z}}
:= \max \{ y \in \pi {\bf Z} \, | \, y \le x \}, 
\quad
(x)_{\pi {\bf Z}} 
:= x - [x]_{\pi{\bf Z}}.
\]
We study
the limit of 
$\xi_L$
under the following assumption

{\it
{\bf (A)}

(1)
$\alpha > \frac 12$, 

(2)
A sequence 
$\{ L_j \}_{j=1}^{\infty}$
satisfies
$\lim_{j \to \infty} L_j = \infty$, 
and 
\[
( \sqrt{E_0} L_j )_{\pi {\bf Z}}
=
\beta + o(1), 
\quad
j \to \infty.
\]

for some 
$\beta \in [0, \pi)$.
}

\noindent
Condition 
A(2)
is set to guarantee the convergence of 
$\xi_L$ 
to a point process. 
If 
$a \equiv 0$
for instance, 
A(2) 
is indeed necessary. 

\begin{theorem}
\label{clock}
Assume (A).
Then 
$\xi_{L_j}$
converges in distribution to the clock process with spacing 
$\pi$
with respect to a probability measure
$\mu_{\beta}$ 
on 
$[0, \pi)$.
\end{theorem}
\begin{remark}
Let 
$x_t$
be the solution to the eigenvalue equation : 
$H_L x_t = \kappa^2 x_t$ ($\kappa > 0$).
Let 
$\tilde{\theta}(\kappa)$
be the one defined in (\ref{thetatilde}). 
Then 
$\tilde{ \theta }_t(\kappa)$
has a limit as 
$t$
goes to infinity\cite{KU} :  
$\lim_{t\to \infty}\tilde{\theta}_t(\kappa) = \tilde{\theta}_{\infty}(\kappa)$,
a.s. ; 
$\mu_{\beta}$
is the distribution of the random variable
$(\beta + \tilde{\theta}_{\infty}(\sqrt{E_0}) )_{\pi {\bf Z}}$.
In some special cases, 
we can show that 
$(\tilde{\theta}_{\infty}(\sqrt{E_0}) )_{\pi {\bf Z}}$
is not uniformly distributed on 
$[0, \pi)$
for large
$E_0$, 
implying that 
$\mu_{\beta}$
really depends on $\beta$. 
\end{remark}
\begin{remark}
We can consider point processes with respect to two reference energies 
$E_0, E'_0 (E_0 \ne E'_0)$ 
simultaneously : suppose a sequence
$\{ L_j \}_{j = 1}^{\infty}$
satisfies
$
( \sqrt{E_0} L_j )_{\pi {\bf Z}}
=
\beta + o(1)$, 
$
( \sqrt{E'_0} L_j )_{\pi {\bf Z}}
=
\beta' + o(1)$, 
$j \to \infty$
for some 
$\beta, \beta' \in [0, \pi)$. 
We set 
$\xi_L := \sum_{n\ge n(L)}
\delta_{L (\sqrt{E_n(L)} - \sqrt{E_0})}$, 
$\xi'_L := \sum_{n \ge n(L)}
\delta_{L (\sqrt{E_n(L)} - \sqrt{E'_0})}$.
Then the joint distribution of 
$\xi_{L_j}, \xi'_{L_j}$
converges, for 
$f, g \in C_c({\bf R})$,
\beq
&&\lim_{j \to \infty}
{\bf E}\left[
\exp \left(
- \xi_{L_j}(f) - \xi_{L_j}(g)
\right)
\right]
\\
&&\qquad
=
\int_0^{\pi} 
d \mu(\phi, \phi')
\exp \left(
- \sum_{n \in {\bf Z}}
( f (n \pi - \phi) + g ( n \pi - \phi') )
\right)
\eeq
where
$\mu(\phi, \phi')$
is the joint distribution of 
$(\beta + \tilde{\theta}_{\infty}(\sqrt{E_0}) )_{\pi {\bf Z}}$
and 
$(\beta' + \tilde{\theta}_{\infty}(\sqrt{E'_0}) )_{\pi {\bf Z}}$.
We are unable to identify 
$\mu(\phi, \phi')$
but it may be possible that 
$\phi$ and  $\phi'$ 
are correlated. 
\end{remark}
\begin{remark}
Suppose we renumber the eigenvalues near the reference energy 
$E_0$ 
so that  
\[
\cdots < E'_{-2}(L) < E'_{-1}(L) < E_0 \le E'_0(L) < E'_1(L) < E'_2(L) < \cdots.
\]
Then an argument similar to the proof of 
Theorem 2.4 in \cite{K} proves the following fact : 
for any 
$n \in {\bf Z}$
we have
\begin{equation}
\lim_{L \to \infty} 
L (\sqrt{E'_{n+1}(L)} - \sqrt{E'_n (L)}) = \pi, 
\quad
a.s.
\label{strong}
\end{equation}
which is called the strong clock behavior \cite{ALS}. 
We note that the integrated density of states is equal to
$\sqrt{E}/\pi$. 
\end{remark}
We next study the finer structure of the eigenvalue spacing, under the following assumption. 

{\it
{\bf (B)}

(1)
$a(t) = t^{-\alpha} (1 + o(1))$, 
$t \to \infty$, 
$\frac 12 < \alpha < 1$, 

(2)
A sequence 
$\{ L_j \}_{j=1}^{\infty}$
satisfies 
$\lim_{j \to \infty}L_j = \infty$
and
\[
\sqrt{E_0} L_j 
=
m_j \pi + \beta + \epsilon_j, 
\quad
j \to \infty
\]

for some $\{ m_j \}_{j=1}^{\infty}(\subset{\bf N})$, 
$\beta \in [0, \pi)$
and 
$\{ \epsilon_j \}_{j=1}^{\infty}$
with 
$\lim_{j \to \infty} \epsilon_j = 0$. 
}

\noindent
Roughly speaking, 
$E_{m_j}(L_j)$
is the eigenvalue closest to 
$E_0$. 
%
In view of 
(\ref{strong}),   
we set 
\[
X_j(n) :=
\left\{
\left(
\sqrt{E_{m_j + n+1}(L_j)} - \sqrt{E_{m_j + n}(L_j)}
\right) L_j 
- \pi
\right\} L_j^{\alpha - \frac 12}, 
\quad
n \in {\bf Z}.
\]
\begin{theorem}
\label{second}
Assume (B). 
Then 
$\{ X_j(n) \}_{n \in {\bf Z}}$
converges in distribution to the Gaussian system with covariance
\beq
C(n, n')
&=&
\frac {C(E_0)}{8E_0}
Re
\int_0^1
s^{- 2 \alpha}
e^{2i (n - n') \pi s}
2(1 - \cos 2\pi s) ds,  
\quad
n, \; n' \in {\bf Z}, 
\eeq
where
$C(E)
:=
\int_M 
\left| 
\nabla (L + 2i \sqrt{E})^{-1}F 
\right|^2 dx$ 
and 
$L$
is the generator of 
$(X_t)$. 
\end{theorem}
\begin{remark}
Lemma 2.1 
in \cite{KU}
and
Lemma \ref{small} 
imply that 
\[
\sqrt{ 
E_{m_j}(L_j)
}
=
\sqrt{E_0}
-
\frac {
\beta + \tilde{\theta}_{\infty}(\sqrt{E_0}) 
}
{L_j}
+
Y_j
\]
where
$Y_j
=
O(L_j^{-\alpha - \frac 12 + \epsilon }
)
+
O(\epsilon_j L_j^{-1})$, 
a.s.
for any 
$\epsilon > 0$. 
Furthermore by definition of 
$\{ X_j (n) \}$ 
we have
\beq
\sqrt{ 
E_{m_j+n}(L_j)
}
&=&
\left\{
\begin{array}{@{\,}ll}
\sqrt{
E_{m_j}(L_j)}
+
\frac {n \pi}{L_j}
+
\frac {1}{L_j^{\alpha+ \frac 12}}
\sum_{l=0}^{n-1} X_j(l)
&
(n \ge 1) \\
\sqrt{
E_{m_j}(L_j)}
+
\frac {n \pi}{L_j}
-
\frac {1}{L_j^{\alpha+ \frac 12}}
\sum_{l=n}^{-1} X_j(l)
&
(n \le -1).\\
\end{array}
\right.
\eeq
Theorem \ref{second}
thus describes the behavior of eigenvalues near 
$E_{m_j}(L_j)$
in the second order.
\end{remark}
\begin{remark}
Suppose we consider two reference energies 
$E_0, E'_0 (E_0 \ne E'_0)$ 
simultaneously and suppose a sequence
$\{ L_j \}_{j = 1}^{\infty}$
satisfies
$\lim_{j \to \infty} L_j = \infty$
and
$\sqrt{E_0} L_j 
=
m_j \pi+ \beta + o(1)$, 
$\sqrt{E'_0} L_j 
=
m'_j \pi+ \beta' + o(1)$, 
$j \to \infty$
for some 
$m_j, m'_j \in {\bf N}$,
and 
$\beta, \beta' \in [0, \pi)$. 
Then 
$\{ X_j(n) \}_n$
and 
$\{X'_j(n) \}_{n}$
converge jointly to the mutually independent Gaussian systems. 
\end{remark}
\subsection{Critical Case}
We set the following assumption. 

{\bf (C)}
$\quad a(t) = t^{- \frac 12} (1 + o(1)), 
\quad
t \to \infty$. 
\begin{theorem}
\label{sc-limit}
Assume (C).
Then
\begin{equation}
\lim_{L \to \infty}
{\bf E}[ e^{- \xi_L(f)} ] 
=
{\bf E}\left[
\int_0^{2 \pi} \frac {d \theta}{2 \pi}
\exp \left(
- \sum_{n \in {\bf Z}}
f ( \Psi_1^{-1} (2 n \pi + \theta) )
\right)
\right]
\label{Laplacebeta}
\end{equation}
where 
$\{ \Psi_t (\cdot) \}_{t \ge 0}$
is the strictly-increasing function valued process such that
for any 
$c_1, \cdots, c_m \in {\bf R}$, 
$\{ \Psi_t (c_j) \}_{j=1}^m$
is the unique solution to the following SDE : 
\begin{eqnarray}
d \Psi_t(c_j) &=& 2c_j dt + 
D(E_0)
Re 
\left\{
(e^{i \Psi_t(c_j)} - 1) 
\frac {d Z_t}{\sqrt{t}}
\right\}
\label{SDEbeta}
\\
\Psi_0(c_j) &=& 0, 
\quad
j=1, 2, \cdots, m
\nonumber
\end{eqnarray}
where
$C(E_0)
:=
\int_M 
\left| 
\nabla (L + 2i \sqrt{E_0})^{-1}F 
\right|^2 dx$, 
$D(E_0) :=\sqrt{
\frac {C(E_0)}{2 E_0}
}$
and
$Z_t$
is a complex Browninan motion.
\end{theorem}
\begin{definition}
For 
$\beta > 0$, 
the circular $\beta$-ensemble with $n$-points is given by
\[
{\bf E}_n^{\beta}[G]
:=
\frac {1}{Z_{n, \beta}}
\int_{-\pi}^{\pi} \frac {d \theta_1}{2 \pi}
\cdots
\int_{-\pi}^{\pi} \frac {d \theta_n}{2 \pi}
G(\theta_1, \cdots, \theta_n)
| \triangle (e^{i \theta_1}, \cdots, e^{i \theta_n}) |^{\beta}
\]
where 
$Z_{n, \beta}$
is the normalization constant, 
$G \in C({\bf T}^n)$
is bounded and 
$\triangle$
is the Vandermonde determinant.
The limit 
$\xi_{\beta}$
of the circular $\beta$-ensemble is defined by
\[
{\bf E}[ e^{- \xi_{\beta}(f)} ]
=
\lim_{n \to \infty}
{\bf E}_n^{\beta} 
\left[
\exp 
\left(
- \sum_{j=1}^n f(n \theta_j)
\right)
\right], 
\quad
f\in C^+_c({\bf R})
\]
\end{definition}
Killip-Stoiciu \cite{KS}
proved that the limit 
$\xi_{\beta}$
exists and satisfies (\ref{Laplacebeta}), (\ref{SDEbeta}) 
where 
$D(E_0)$
is replaced by 
$\frac {2}{\sqrt{\beta}}$ 
and
$2 c_j$
is replaced by 
$c_j$. 
Therefore the limit of 
$\xi_L$
coincides with that of the circular $\beta$-emsemble modulo a scaling.

\begin{corollary}
\label{critical}
Assume (C).
Writing  
$\xi_{\beta} = \sum_n \delta_{\lambda_n}$,
let
$\xi'_{\beta} := \sum_n \delta_{\lambda_n/2}$.
Then 
$\xi_{L} \stackrel{d}{\to} \xi'_{\beta}$
with
$\beta = \beta (E_0) :=\frac {8 E_0}{C(E_0)}$. 
\end{corollary}
\begin{remark}
The corresponding 
$\beta = \beta (E_0) =\frac {8 E_0}{C(E_0)}$
depends on the reference energy 
$E_0$, 
so that the spacing distribution may change if we look at the different region in the spectrum.
In fact we have 
$\beta(E) = \gamma(E)^{-1}$ 
where 
$\gamma(E)$
is the Lyapunov exponent defined in 
\cite{KU} 
such that the generalized eigenfunction 
$\psi_E$
of 
$H$
satisies 
$\psi_E \simeq |x|^{-\gamma(E)}$, 
$|x| \to \infty$. 
It then follows that 
$E < E_c$
(resp. $E > E_c$)
if and only if 
$\beta(E) < 2$
(resp. $\beta(E) > 2$)
and 
$\beta(E_c) = 2$ 
(Figure 1.).
Similar statement 
also holds for discrete Hamiltonian, the Jacobi matrix arising from the $\beta$-ensemble, and CMV matrices studied respectively by \cite{KVV, BFS, KS}. 
This is consistent 
with our general belief that in the 
point spectrum (resp. in the continuous spectrum)
the level repulsion is 
weak (resp. strong).

\vspace*{1em}

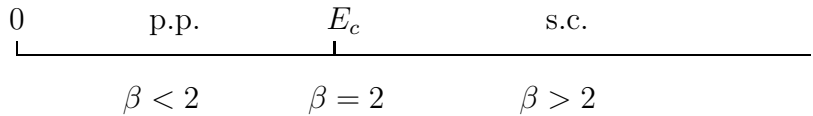
\begin{figure}[h]
\begin{center}
\begin{picture}(300, 30)
\put(0, 10){

\put(0,0){\line(1,0){300}}
\put(0,0){\line(0,1){5}}
\put(120,0){\line(0,1){5}}
\put(-3, 10){$0$}
\put(117, 10){$E_c$}
\put(50,10){p.p.}
\put(200, 10){s.c.}
\put(40,-20){$\beta<2$}
\put(190, -20){$\beta>2$}
\put(110, -20){$\beta=2$}

}
\end{picture}
\end{center}
\caption{
Spectrum and corresponding $\beta$.
\label{Figure 1.}}
\end{figure}
\vspace*{1em}

We note that, for
$\beta=2$, 
the circular 
$\beta$-ensemble with 
$n$-points 
coincides with the eigenvalue distribution of the unitary ensemble with the Haar measure on 
$U(n)$. 
In \cite{VV}, 
Valk\'o-Vir\'ag showed that 
Sine$_{\beta}$ process has a phase transition at 
$\beta = 2$. 
\end{remark}
\begin{remark}
In 
\cite{Nakano}, 
it is proved that 
$\xi_L$
also converges to the limit of the Gaussian 
$\beta$-ensemble for same 
$\beta$, 
thus proving the coincidence of limits of two 
$\beta$-ensembles. 
\end{remark}
\begin{remark}
If we consider two reference energies 
$E_0, E'_0 (E_0 \ne E'_0)$, 
then the corresponding point process 
$\xi_L, \xi'_L$
converges jointly to the independent 
$\xi_{\beta}, \xi'_{\beta'}$. 
\end{remark}
In later sections, 
we prove theorems mentioned above based on the argument in 
\cite{KS, KU, K} : 
The main ingredient of the proof is to study the limiting behavior of the relative Pr\"ufer phase
$\Psi_L$, 
by which the Laplace transform of 
$\xi_L$
is represented(Lemma \ref{Laplace}). 
The major difference
from the argument in \cite{KS} is that 
$\phi(E_0, L)$, 
which is defined in Section 2 
to be the projection to the torus 
of the Pr\"ufer phase associated to 
$E_0$, 
is not uniformly distributed and is not independent of 
$\Psi_L$. 
Hence 
our additional task is to show that, the joint limit of 
$(\Psi_L, \phi(E_0, L))$
is independent each other, 
the convergence of $\Psi_L$ is stronger, and that the limit of which is strictly monotone and continuous. 
In Section 2
we prepare some notations and basic facts. 
In Sections 3, 4, 
we consider the ac-case and prove
Theorems \ref{clock}, \ref{second}. 
In Sections 6-9, 
we consider the critical case and prove Theorem \ref{sc-limit}
which is outlined in Section 5.
In what follows, 
$C$
denotes general positive constant which is subject to change from line to line in each argument.
%

\section{Preliminaries}
Let 
$x_t$
be the solution to the equation 
$H_L x_t = \kappa^2 x_t$
$(\kappa > 0)$
which we set in the following form 
\begin{equation}
\left( \begin{array}{c}
x_t \\ x'_t /\kappa
\end{array} \right)
=
r_t
\left( \begin{array}{c}
\sin \theta_t \\ \cos \theta_t
\end{array} \right), 
\quad
\theta_0 = 0.
\label{Prufer}
\end{equation}
We define 
$\tilde{\theta}_t(\kappa)$
by 
\begin{equation}
\theta_t (\kappa) = \kappa t + \tilde{\theta}_t (\kappa).
\label{thetatilde}
\end{equation}
Then it follows that 
\begin{eqnarray}
r_t(\kappa)
&=&
\exp \left(
\frac {1}{2\kappa} Im 
\int_0^t 
a(s) F(X_s) e^{2i \theta_s(\kappa)} ds
\right)
\label{r-eq}
\\
\tilde{\theta}_t (\kappa)
&=&
\frac {1}{2 \kappa}
\int_0^t
Re (e^{2i \theta_s(\kappa)} -1 ) a(s) F(X_s)
\label{diff-eq}
\\
\frac {\partial \theta_t(\kappa)}{\partial \kappa}
&=&
\int_0^t 
\frac {r_s^2}{r_t^2} ds
+
\frac {1}{2 \kappa^2}
\int_0^t 
\frac {r_s^2}{r_t^2}
a(s) F(X_s) 
(1 - Re \; e^{2i \theta_s(\kappa)}) ds. 
\label{theta-kappa}
\end{eqnarray}
By using 
the behavior of solutions 
$x_t$ 
\cite{KU}, 
we can show the following fact : let  
$I \subset (0, \infty)$
be an interval. 
Then for sufficiently large 
$t > 0$ 
we have 
$\inf_{\kappa \in I}
\frac {\partial \theta_t(\kappa)}{\partial \kappa} >0$, 
so that 
$\theta_t(\kappa)$
is increasing as a function of 
$\kappa$
on 
$I$. 
Here and henceforth, for simplicity, we say 
$f$
is increasing if and only if 
$x<y$ 
implies 
$f(x) < f(y)$. 
Set
\begin{eqnarray}
m(E_0, L)\pi &:=& [\theta_L(\sqrt{E_0})]_{\pi {\bf Z}},
\quad
\phi(E_0, L) := (\theta_L (\sqrt{E_0}))_{\pi {\bf Z}}
\in [0, \pi).
\label{rep}
\end{eqnarray}
Moreover we define the relative Pr\"ufer phase
\[
\Phi_L (x)
=
\theta_L(\sqrt{E_0}+\frac {x}{L})
-
\theta_L(\sqrt{E_0})
\]
which is continuous and increasing.
As in \cite{KS}
we use the following representation of Laplace transform of 
$\xi_L$
in terms of 
$\Phi_L$. 
\begin{lemma}
\label{Laplace}
For
$f \in C^+_c({\bf R})$
we have
\[
{\bf E}[e^{- \xi_L(f)}]
=
{\bf E}\left[
\exp \left(
-\sum_{n =n(L) - m(E_0, L)}^{\infty}
f\Bigl( \Phi_L^{-1} (n \pi - \phi(E_0, L) \Bigr)
\right)
\right].
\]
\end{lemma}
%
%
\section{Convergence to a clock process}
In what follows, for simplicity, we set 
\[
\kappa := \sqrt{E_0}
\]
%
\subsection{The behavior of $\Psi_L$}
\begin{proposition}
\label{Psi}
If 
$\alpha > \frac 12$, 
following fact holds for a.s. : 
\[
\lim_{L \to \infty}
\Phi_L (x) = x
\]
pointwise and this holds compact uniformly with respect to
$\kappa$.
\end{proposition}
\begin{proof}
By
(\ref{diff-eq})
we have
\begin{eqnarray*}
\Phi_L(x)
&=&
x + 
\frac {1}{2 \kappa}
Re \, 
\int_0^L a(s) F(X_s) 
\left(
e^{2i \theta_s(\kappa+\frac xL)}
-
e^{2i \theta_s(\kappa)}
\right) ds
+
O(L^{-\alpha}). 
\end{eqnarray*}
We set 
\[
A_t(\kappa, \beta) :=
\int_0^t a(s) F(X_s) e^{i \beta \theta_s(\kappa)} ds.
\]
Take
$\delta > 0$ such that 
$\int_0^{\infty} a(s)^2 s^{\delta} ds < \infty$.
Then by 
\cite{KU} Lemma 2.2,  
for any compact set 
$K \subset (0, \infty)$ 
and for any 
$\epsilon < \frac {\delta}{2}$, 
$\beta \in {\bf R}$
we have 
\[
\sup_{t \ge 0, \;\kappa, \;\kappa_1 \in K}
\frac {
|A_t(\kappa, \beta) - A_t(\kappa_1, \beta)|
}
{ | \kappa - \kappa_1 |^{\epsilon} } < \infty, 
\quad
a.s..
\]
Hence for fixed $x$, we have 
$\Phi_L(x) = x + O(L^{-\epsilon})$, 
a.s..
Since the function 
$f(x) = x$
is continuous, the proof is complete.
\QED
\end{proof}
\subsection{Proof of Theorem \ref{clock}}
We sometimes use the following elementary lemma. 
\begin{lemma}
\label{inverse}
Let 
$\Psi_n, n=1,2, \cdots$, 
and 
$\Psi$
are continuous and increasing functions on an open interval 
$I$ 
such that
$\lim_{n \to \infty}\Psi_n(x)=\Psi(x)$ 
pointwise. 
If 
$y_n \in Ran \,\Psi_n$, 
$y \in Ran \,\Psi$
and
$y_n \to y$, 
then it holds that 
$
\Psi_n^{-1}(y_n) 
\stackrel{n \to \infty}{\to}
\Psi^{-1}(y).
$
\end{lemma}
{\it Proof of Theorem \ref{clock}}\\
By the fact that 
$\tilde{\theta}_t(\kappa) \stackrel{t\to\infty}{=} \tilde{\theta}_{\infty}(\kappa) + o(1)$
(\cite{KU} Proposition 2.1)
and by (A)(2), 
$\lim_{j \to \infty}
\phi(\kappa^2, L_j)
=
\left(
\tilde{\theta}_{\infty}(\kappa)+\beta 
\right)_{\pi {\bf Z}}$, 
a.s..
Together with Proposition \ref{Psi},  
the assumption for 
Lemma \ref{inverse}
is satisfied.
\QED
%
\section{Second Limit Theorem}
\subsection{Behavior of eigenvalues near $E_0$}
\begin{lemma}
\label{small}
Assume 
(B)
and let 
$n \in {\bf Z}$. 
Then for 
$j \to \infty$
we have 
\beq
(1)&& \qquad
\sqrt{E_{m_j+n}(L_j)}
=
\kappa+ o(1)
\\
(2) && \qquad
\sqrt{E_{m_j+n}(L_j)}
=
\kappa
+
\frac {
n \pi - \beta - \tilde{\theta}_{\infty}(\kappa)}
{L_j}
+
o(L_j^{-1}).
\eeq
\end{lemma}
Lemma \ref{small}
follows from the fact that 
$\tilde{\theta}_L (\kappa) \stackrel{L \to \infty}{\to} \tilde{\theta}_{\infty}(\kappa)$
holds compact uniformly w.r.t. 
$\kappa$. 
By definition we see that 
\beq
X_j(n)
&=&
-L_j^{\alpha - \frac 12}
\left(
\tilde{\theta}_{L_j} (\sqrt{E_{m_j+n+1}(L_j)})
-
\tilde{\theta}_{L_j} (\sqrt{E_{m_j+n}(L_j)})
\right).
\eeq
By
Lemma \ref{small}(2)
\beq
&&
\sqrt{E_{m_j+n+1}(L_j)} 
= \kappa + \frac {c_1}{L_j}, 
\quad
\sqrt{E_{m_j+n}(L_j)} = \kappa + \frac {c_2}{L_j}
\\
&&
c_1 = (n+1) \pi - \beta - \tilde{\theta}_{\infty}(\kappa)+o(1),
\quad
c_2 = n \pi - \beta - \tilde{\theta}_{\infty}(\kappa)+o(1), 
\quad
j \to \infty.
\eeq
We set 
\beq
\Theta_t^{(n)}(c_1,c_2)
&:=&
\left(
\tilde{\theta}_{nt}(\kappa + \frac {c_1}{n})
-
\tilde{\theta}_{nt}(\kappa + \frac {c_2}{n})
\right)
n^{\alpha - \frac 12}
\\
l_t(
(c_1, c_2), (c'_1, c'_2)
)
&:=&
\frac {C(\kappa^2)}{8 \kappa^2}
\int_0^t
s^{-2 \alpha}
Re 
\left(
e^{2i c_1 s} - e^{2i c_2 s}
\right)
\overline{
\left(
e^{2i c'_1 s} - e^{2i c'_2 s}
\right)
}
ds.
\eeq
When 
$c_1, c_2$
are constant, the following fact is proved in \cite{K} Lemma 3.1.
\begin{proposition}
$\{ \Theta^{(n)}_t(c_1, c_2) \}_{t \ge 0, \,c_1, c_2 \in {\bf R}}
\stackrel{d}{\to}
\{Z(t, c_1, c_2) \}_{t \ge 0, \,c_1, c_2 \in {\bf R}}$
as 
$n \to \infty$
where
$\{Z(t, c_1, c_2) \}_{t \ge 0, \,c_1, c_2 \in {\bf R}}$
is the Gaussian system with covariance
$l_{t \wedge t'}((c_1, c_2), (c'_1, c'_2))$.
\end{proposition}
\subsection{Independence of the limits}
To finish the proof of Theorem \ref{second}, 
it is then sufficient to prove that 
$(\tilde{\theta}_{nt}(\kappa), \{ \Theta_t^{(n)} (c_1, c_2)) \}_{c_1, c_2})$
converges jointly to the independent ones. 
Let 
$0 < \kappa_1 < \kappa_2$
and 
$I:=[\kappa_1, \kappa_2]$. 
In the lemma below, 
we regard
$\tilde{\theta}_t, \tilde{\theta}_{\infty}$
are 
$C(I)$-valued random elements. 
\begin{lemma}
For 
$t > 0$
fixed, we have 
\[
(\tilde{\theta}_{nt}, \{ \Theta_t^{(n)} (c_1, c_2) \}_{c_1, c_2})
\stackrel{d}{\to}
(\tilde{\theta}_{\infty}, \{ Z(t, c_1, c_2) \}_{c_1, c_2})
\]
as
$n \to \infty$
where
$\tilde{\theta}_{\infty}$
and 
$\{ Z(t, c_1, c_2) \}_{c_1, c_2}$
are independent. 
\end{lemma}
\begin{proof}
Let 
$A (\subset C(I))$
be a 
$\tilde{\theta}_{\infty}$-continuity set
(i.e.,  
${\bf P}(\tilde{\theta}_{\infty} \in \partial A) = 0$)
and set 
$A_{\epsilon} := \{ f \in C(I) \, | \, 
d(f,A) < \epsilon \}$.
Since 
$\tilde{\theta}_t(\kappa) \stackrel{a.s.}{\to} \tilde{\theta}_{\infty}(\kappa)$
compact uniformly in 
$\kappa$, 
for any 
$\epsilon > 0$
${\bf P}\left(
\tilde{\theta}_{nt} \in A, 
\;
\tilde{\theta}_T \notin A_{\epsilon} 
\right)
=
o(1)$
for sufficiently large 
$T, n$. 
Here we recall eq.(3.3) in 
\cite{K}.
\beq
\Theta_t^{(n)}(c_1, c_2)
&=&
T_t^{(n)}(c_1, c_2)
+ 
O(n^{\frac 12 - \alpha})
\\
\mbox{ where }\quad
T_t^{(n)}(c_1, c_2)
&:=&
n^{\alpha - \frac 12}
Re 
\left(
S_t^{(n)}
\left( \kappa + \frac {c_1}{n} \right)
 - 
S_t^{(n)}
\left( \kappa + \frac {c_2}{n} \right)
\right)
\\
S_t^{(n)}(\kappa)
&:=&
\frac {1}{2 \kappa}
\int_0^{nt}
a(s) e^{2i \tilde{\theta}_s(\kappa)}
d M_s(\kappa)
\eeq
$M_s(\kappa)$
is the complex martingale defined in subsection 6.2.
Let 
$m \in {\bf N}$.
For
${\bf c}_1
=
(c^{(1)}_1, \cdots, c^{(m)}_1)$, 
${\bf c}_2
=
(c^{(1)}_2, \cdots, c^{(m)}_2)$, 
we use the following convention : 
$\Theta_t^{(n)} ({\bf c}_1, {\bf c}_2)=
\left(
\Theta_t^{(n)}(c^{(1)}_1, c^{(1)}_2), 
\cdots,
\Theta_t^{(n)}(c^{(m)}_1, c^{(m)}_2)
\right)$
and similarly for 
$T_t^{(n)}({\bf c}_1, {\bf c}_2)$
and 
$Z(t, {\bf c}_1, {\bf c}_2)$.
Let 
$B \in {\cal B}({\bf R}^m)$
be a 
$Z(t, {\bf c}_1, {\bf c}_2)$-continuity set and let 
$B_{\epsilon} := \{ x \in {\bf R}^m \, | \, d(x,B) < \epsilon \}$. 
Writing 
$\Theta^{(n)}_t = \Theta^{(n)}_t({\bf c}_1, {\bf c}_2)$, 
$T^{(n)}_t = T^{(n)}_t({\bf c}_1, {\bf c}_2)$
we have, for sufficiently large 
$n$, 
\beq
{\bf P}\left(
\tilde{\theta}_{nt} \in A, 
\Theta_t^{(n)} \in B
\right)
& \le &
{\bf P}\left(
\tilde{\theta}_T \in A_{\epsilon}, 
T_t^{(n)} \in B_{\epsilon}
\right)
+ o(1)
\\
&=&
{\bf P}\left(
\tilde{\theta}_T \in A_{\epsilon}, 
T_t^{(n)} - T_{T/n}^{(n)} + T_{T/n}^{(n)}
\in B_{\epsilon}
\right)
+ o(1)
\\
&=&
{\bf P}\left(\tilde{\theta}_{T} \in A_{\epsilon}, 
T_t^{(n)}- T_{T/n}^{(n)} \in B_{2\epsilon}
\right)
+ o(1).
\eeq
Here we used 
$T^{(n)}_{T/n} \stackrel{P}{\to} 0$.
By the Markov property
\beq
&=&
{\bf E}\left[ 
1_{\{\tilde{\theta}_T \in A_{\epsilon} \}} 
{\bf E}_{X_T}\Bigl[
1_{ \{ \widetilde{T}^{(n)}_{t - T/n} \in B_{2\epsilon} \} } 
\Bigr]
\right] + o(1)
\eeq
where
$\widetilde{T}_t^{(n)}$
is the suitable ``time-shift" of 
$T_t^{(n)}$. 
Because 
$\widetilde{T}_t^{(n)}$
converges in distribution to  
$Z(t, {\bf c}_1, {\bf c}_2)$
as 
$n \to \infty$ 
being irrespective of
$X_T$, 
\beq
&=&
{\bf P}\left(\tilde{\theta}_T \in A_{\epsilon}
\right)
{\bf P}\left(
Z(t, {\bf c}_1, {\bf c}_2) \in B_{2\epsilon}
\right)+ o(1)
\\
& \le &
{\bf P}\left(
\tilde{\theta}_{\infty} \in A_{2\epsilon} 
\right)
{\bf P}\left(
Z(t, {\bf c}_1, {\bf c}_2) \in B_{2\epsilon}
\right)+ o(1).
\eeq
Since 
$A$
is a 
$\tilde{\theta}_{\infty}$-continuity set and 
$B \in {\cal B}({\bf R}^m)$
is a 
$Z(t, {\bf c}_1, {\bf c}_2)$-continuity set, 
\[
\limsup_{n \to \infty}
{\bf P}\left(
\tilde{\theta}_{nt} \in A, \Theta_t^{(n)} \in B
\right)
\le
{\bf P}(\tilde{\theta}_{\infty} \in A)
{\bf P}(Z(t, {\bf c}_1, {\bf c}_2) \in B).
\]
The opposite inequality can be proved similarly.
\QED
\end{proof}
%

\section{SC-case : outline of proof of Theorem \ref{sc-limit}}
In this section
we overview the proof of Theorem \ref{sc-limit}. 
First of all, set 
\beq
&&
[x]_{2\pi {\bf Z}} 
:= \max \{ y \in 2\pi {\bf Z}\, | \, y \le x \}, 
\quad
(x)_{2\pi{\bf Z}} := x - [x]_{2\pi {\bf Z}}, 
\\
&&
2 m(\kappa^2, L)\pi := [ 2 \theta_L(\kappa)]_{2 \pi {\bf Z}}, 
\quad
\phi(\kappa^2, L) := (2\theta_L (\kappa))_{2\pi {\bf Z}}
\in [0, 2\pi).
\eeq
We also set the relative Pr\"ufer phase by 
$
\Psi_L (x)
:=
2\theta_L(\kappa+\frac {x}{L})
-
2\theta_L(\kappa).
$
Then we have a variant of 
Lemma \ref{Laplace}.
\begin{lemma}
\label{Laplace2}
For 
$f \in C^+_c({\bf R})$
\[
{\bf E}[e^{- \xi_L(f)}]
=
{\bf E}\left[
\exp \left(
-\sum_{n =n(L) - m(\kappa^2, L)}^{\infty}
f\Bigl( \Psi_L^{-1} (2n \pi - \phi(\kappa^2, L)) \Bigr)
\right)
\right].
\]
\end{lemma}
So our task 
is to study the limit of the joint distribution of 
$(\Psi_L, \phi(\kappa^2, L))$
as 
$L \to \infty$.
Following 
\cite{KS}
we consider
\begin{equation}
\Psi_t^{(n)}(x)
:=
2\theta_{nt}(\kappa+\frac {x}{n})
-
2\theta_{nt}(\kappa),
\label{Psiscaling}
\end{equation}
regard it as an increasing function-valued process, and find a process 
$\Psi_t(x)$
such that 
for any fixed 
$c_1, \cdots, c_m \in {\bf R}$
$\{ \Psi_t^{(n)}(c_j) \}_{j=1}^m
\stackrel{d}{\to} \{ \Psi_t(c_j) \}_{j=1}^m$
(Theorem \ref{SDE}).
$\Psi_t$
is characterized as the unique solution to the SDE (\ref{SDEbeta}). 
Moreover, 
$\Psi_t(c)$
is continuous and increasing with respect to 
$c$
(Lemma \ref{parameter-continuity}). 
On the other hand
we have 
$(\{ \Psi_1^{(n)}(c_j) \}_{j=1}^m, \phi(\kappa^2, n))
\stackrel{d}{\to}
(\{ \Psi_1(c_j) \}_{j=1}^m, \phi_1)$
jointly, where 
$\phi_1$
is uniformly distributed on 
$[0, 2 \pi)$
and independent of 
$\Psi_1$
(Proposition \ref{joint-convergence}). 
Moreover
$\Psi^{(n)}$
converges to 
$\Psi$
also as a sequence of increasing function-valued process 
(Lemma \ref{compactness}), 
so that we can find a coupling such that 
for a.s. 
$( 
(\Psi^{(n)}_1)^{-1}(x), \phi(\kappa^2, n)
)
\to
(\Psi_1^{-1}(x), \phi_1)$
for any 
$x \in {\bf R}$
(Proposition \ref{coupling}). 
Therefore we obtain (\ref{Laplacebeta}). 
%
\section{Convergence of 
$\Psi$}
\subsection{Preliminaries}
We recall 
the basic tool used in \cite{KU, K}. 
For 
$f \in C^{\infty}(M)$
let 
$R_{\beta} f 
:=(L + i\beta)^{-1}f$
$(\beta > 0)$,  
$R_0 f :=L^{-1}(f - \langle f \rangle)$.
Then by 
Ito's formula, 
\beq
\int_0^t e^{i \beta s} f(X_s) ds
&=&
\left[
e^{i \beta s} (R_{\beta}f)(X_s) 
\right]_0^t
+
\int_0^t 
e^{i \beta s} d M_s(f, \beta)
\\
\int_0^t f(X_s) ds
&=&
\langle f \rangle t
+
\left[ (R_0 f) (X_s) \right]_0^t 
+
M_t(f,0)
\eeq
where
$M_s(f, \beta), M_s(f, 0)$
are the complex martingales whose variational process satisfy
\beq
\langle M(f, \beta), M(f, \beta) \rangle_t
&=&
\int_0^t [ R_{\beta} f, R_{\beta} f] (X_s) ds, 
\\
\langle M(f, \beta), \overline{M(f, \beta)} \rangle_t
&=&
\int_0^t [ R_{\beta} f, \overline{R_{\beta} f}] (X_s) ds
\eeq
where
\beq
[f_1, f_2](x)
:=
L(f_1 f_2)(x)
-
(Lf_1)(x) f_2(x)
-
f_1(x) (Lf_2)(x)
=
(\nabla f_1, \nabla f_2) (x).
\eeq
Then 
the integration by parts gives us the following formulas to be used frequently. 
\begin{lemma}
\label{partial integration}
\beq
(1)
&&
\int_0^t 
b(s) e^{i \beta  s}
e^{i \gamma \tilde{\theta}_s}
f(X_s) ds
\\
&=&
\left[
b(s) e^{i \gamma \tilde{\theta}_s}
e^{i \beta  s}
(R_{\beta }f)(X_s)
\right]_0^t
- 
\int_0^t 
b'(s) e^{i \gamma \tilde{\theta}_s}
e^{i \beta  s}
(R_{\beta }f)(X_s)ds
\\
&& - 
\frac {i \gamma}{2 \kappa}
\int_0^t 
b(s) a(s) Re (e^{2i \theta_s}-1) 
e^{i \gamma \tilde{\theta}_s}
e^{i \beta  s}
F(X_s) (R_{\beta }f) (X_s) ds
\\
&& +
\int_0^t b(s) 
e^{i \beta  s} 
e^{i \gamma \tilde{\theta}_s}
d M_s(f, \beta). 
\eeq
\beq
(2) &&
\int_0^t b(s) e^{i \gamma \tilde{\theta}_s} f(X_s) ds
\\
&=&
\langle f \rangle 
\int_0^t b(s) e^{i \gamma \tilde{\theta}_s} ds
\\
&& + \left[
b(s) e^{i \gamma \tilde{\theta}_s} 
(R_0f)(X_s) 
\right]_0^t 
 - \int_0^t
b'(s) e^{i \gamma \tilde{\theta}_s}
(R_0 f)(X_s) ds
\\
&& - \frac {i \gamma}{2 \kappa}
\int_0^t a(s)b(s) Re (e^{2i \theta_s} - 1)
e^{i \gamma \tilde{\theta}_s}
F(X_s) (R_0 f)(X_s) ds
\\
&& +
\int_0^t b(s) e^{i \gamma \tilde{\theta}_s}
dM_s(f, 0).
\eeq
\end{lemma}
We will also use following notation for simplicity. 
\beq
g_{\kappa} &:=& (L+2i \kappa)^{-1}F, 
\quad
g := L^{-1}(F - \langle F \rangle),
\\
M_s(\kappa) &:=& M_s (F, 2 \kappa), 
\quad
M_s := M_s (F,  0).
\eeq
%
\subsection{A priori estimates}
In this section 
we derive a priori estimate for 
(\ref{Psiscaling}). 
We set 
\beq
Y_t (\kappa)
&:=&
\int_0^t a(s) e^{2i \theta_s(\kappa)} d M_s(\kappa),
\\
\delta_t(\kappa)
&:=&
\left[
a(s) e^{2i \theta_s(\kappa)} g_{\kappa}(X_s)
\right]_0^t 
 -
\int_0^t a'(s) e^{2i \theta_s(\kappa)} g_{\kappa}(X_s) ds
\\
&& \qquad - 
\frac {i}{\kappa}
\int_0^t a(s)^2 e^{2i \theta_s(\kappa)}
\left(
\frac {e^{2i \theta_s(\kappa)}}{2} - 1
\right)
g_{\kappa}(X_s) F(X_s) ds,
\\
V_t^{(n)}(c )
&:=&
Y_{nt}\left(
\kappa + \frac cn
\right) - Y_{nt}(\kappa).
\eeq
\begin{lemma}
\label{fourth}
Suppose
$
\int_0^{\infty} a(s)^3 ds < \infty 
$.
We then have\\
(1)
\[
\int_0^t a(s) e^{2i \theta_s(\kappa)} F(X_s) ds
=
- \frac {i}{2 \kappa}
\int_0^t a(s)^2 g_{\kappa}(X_s) F(X_s) ds
+
Y_t (\kappa)
+
\delta_t(\kappa)
\]
\noindent
(2)
For a.s., 
$\delta_t(\kappa)$
has the limit as 
$t \to \infty$
：
$\lim_{t \to \infty} \delta_t(\kappa) = \delta_{\infty}(\kappa)$, a.s.\\
(3)
For any 
$0 < T < \infty$, 
we have
\[
{\bf E}
\left[ 
\max_{0 \le t \le T}
\left|
\delta_{nt}(\kappa+ \frac cn) - \delta_{nt}(\kappa)
\right|^2
\right]
\stackrel{n \to \infty}{\to} 0.
\]
\end{lemma}
\begin{proof}
(1)
It follows directly from 
Lemma \ref{partial integration}(1).\\
(2)
We further decompose the remainder term 
$\delta_t(\kappa)$ :
\begin{eqnarray}
\delta_t(\kappa)
&=&
\delta^{(1)}_t(\kappa)
+
\delta^{(2)}_t(\kappa)
\label{delta}
\\
\delta_t^{(1)}(\kappa)
&:=&
\left[
a(s) e^{2i \theta_s(\kappa)} g_{\kappa}(X_s) 
\right]_0^t 
- 
\int_0^t a'(s) e^{2i \theta_s(\kappa)} g_{\kappa}(X_s) ds
\label{delta-1}
\\
\delta_t^{(2)}(\kappa)
&:=& - 
\frac {i}{\kappa} 
\int_0^t 
a(s)^2 
\left(
\frac {e^{2i \theta_s(\kappa)}}{2} - 1
\right)
e^{2i \theta_s(\kappa)} 
g_{\kappa}(X_s) F(X_s) ds.
\nonumber
\end{eqnarray}
It is easy to see 
$
\lim_{t \to \infty}
\delta_t^{(1)} (\kappa) = \delta_{\infty}^{(1)}(\kappa)$, 
a.s..
To see the convergence of
$\delta_t^{(2)}(\kappa)$
we write
\begin{eqnarray}
\delta_t^{(2)}(\kappa)
&=&
- \frac {i}{2 \kappa} D_t^{(4)}(\kappa)
+ \frac {i}{\kappa} D_t^{(2)}(\kappa)
\label{delta2}
\\
D_t^{(\beta)}(\kappa)
&:=&
\int_0^t a(s)^2 
e^{i \beta \theta_s (\kappa)} 
F(X_s) g_{\kappa}(X_s) ds, 
\quad
\beta = 2,4.
\nonumber
\end{eqnarray}
We 
use Lemma \ref{partial integration}(1) 
to decompose 
$D_t^{(\beta)}(\kappa)$
into martingale part and the remainder :
Setting  
$h_{\kappa, \beta}= R_{\beta \kappa} (F g_{\kappa})$
and 
$\widetilde{M_s}^{(\beta)} (\kappa) = M_s(F g_{\kappa}, \beta \kappa)$, 
we have
\begin{eqnarray}
&&
D_t^{(\beta)}(\kappa)
=
I_t^{(\beta)}(\kappa) + N_t^{(\beta)}(\kappa)
\label{decomposition}
\\
&&
I_t^{(\beta)}(\kappa)
:=
\left[
a(s)^2 e^{i \beta \theta_s(\kappa)}
h_{\kappa, \beta}(X_s)
\right]_0^t
-
\int_0^t ( a(s)^2 )' 
e^{i \beta \theta_s(\kappa)}
h_{\kappa, \beta}(X_s) ds
\nonumber
\\
&& \qquad\qquad-
\frac {i \beta}{2 \kappa}
\int_0^t a(s)^3
Re (e^{2i \theta_s(\kappa)}-1) 
e^{i \beta \theta_s(\kappa)}
F(X_s) h_{\kappa, \beta}(X_s) ds
\nonumber
\\
&&
N_t^{(\beta)}(\kappa)
:=
\int_0^t a(s)^2
e^{i \beta \theta_s(\kappa)}
d \widetilde{M_s}^{(\beta)} (\kappa).
\nonumber
\end{eqnarray}
$I_t^{(\beta)}(\kappa)$
is easily seen to be convergent : 
$\lim_{t \to \infty}
I_t^{(\beta)}(\kappa) = I_{\infty}^{(\beta)}(\kappa)$, a.s..
Since 
\[
| \langle N^{(\beta)}, N^{(\beta)} \rangle_t|, 
\quad
| \langle N^{(\beta)}, \overline{N^{(\beta)}} \rangle_t |
\le
(const.) \int_0^t a^4(s) ds
< \infty.
\]
Re $N$, Im $N$
can be represented by the time-change of a Brownian motion and thus have limit a.s..\\
(3)
We consider 
$\delta^{(1)}_t(\kappa)$, 
$\delta^{(2)}_t(\kappa)$
separately.
For 
$\delta_t^{(1)}(\kappa)$, we have 
\begin{eqnarray}
& \delta^{(1)}_{nt} &(\kappa+\frac cn)
-
\delta^{(1)}_{nt} (\kappa)
=
a(nt)
\left(
e^{2i \theta_{nt}(\kappa+\frac cn)}
-
e^{2i \theta_{nt}(\kappa)}
\right)
g_{\kappa+\frac cn}(X_{nt})
\nonumber
\\
&&
-
\int_0^{nt} a'(s)
\left(
e^{2i \theta_s(\kappa+\frac cn)} - e^{2i\theta_s(\kappa)}
\right)
g_{\kappa+\frac cn}(X_s) ds
+ O(n^{-1})
\label{diff-delta-1}
\end{eqnarray}
by (\ref{delta-1}).
The second term of (\ref{diff-delta-1}) is 
$o(1)$
as 
$n \to \infty$
due to Lebesgue's dominated convergence theorem. 
For the first term, we note 
\begin{equation}
\max_{0 \le t \le M}
|e^{2i \theta_{t}(\kappa+\frac cn)} - 
e^{2i \theta_t(\kappa)}|
\le \frac {C_M}{n}
\label{M}
\end{equation}
for some positive constant 
$C_M$
depending on 
$M$, 
which follows from  
(\ref{r-eq})-(\ref{theta-kappa}).
We can then show that 
the first term of (\ref{diff-delta-1}) vanishes uniformly w.r.t. 
$t \in [0, T]$ 
so that 
$\max_{0 \le t \le T}{\bf E}[
|
\delta_{nt}^{(1)}(\kappa + \frac cn)
-
\delta_{nt}^{(1)}(\kappa)
|^2
]
\stackrel{n \to \infty}{\to}0$.
Similar argument shows
$\max_{0 \le t \le T} 
\left|
I^{(\beta)}_{nt}(\kappa+\frac cn)
-
I^{(\beta)}_{nt}(\kappa)
\right|
\to 0$
so that we have only to show 
\[
{\bf E}
\left[
\max_{0 \le t \le T}
\left|
N^{(\beta)}_{nt}(\kappa + \frac cn)
-
N^{(\beta)}_{nt}(\kappa)
\right|^2
\right]
\stackrel{n \to \infty}{\to} 0, 
\quad
\beta =2, 4
\]
to finish the proof of Lemma \ref{fourth}(3). 
By the 
martingale inequality,
\beq
&&
{\bf E}
\left[
\max_{0 \le t \le T}
\left|
N^{(\beta)}_{nt}(\kappa + \frac cn)
-
N^{(\beta)}_{nt}(\kappa)
\right|^2
\right]
\le 
C
{\bf E}\left[
\int_0^{nt}
a(s)^4
\left[
H_{\beta, \kappa}, \overline{H_{\beta, \kappa}}
\right]
ds
\right]
\\
&&
\mbox{ where }\quad
H_{\beta, \kappa}(s)
:=
e^{i \beta \theta_s(\kappa + \frac cn)}h_{\beta, \kappa + \frac cn}
-
e^{i \beta \theta_s(\kappa)} h_{\beta, \kappa}
\eeq
which converges to $0$
due to the fact that 
$\int_0^{\infty} a(s)^4 ds < \infty$
and 
Lebesgue's theorem.
\QED
\end{proof}
We assume in what follows
$
a(t) = t^{- 1/2} (1 + o(1)).
$

\begin{lemma}
\label{Theta}
\begin{eqnarray}
\Psi_t^{(n)}(c)
&=&
2ct + 
Re \;\epsilon_t^{(n)}
+
\frac {1}{ \kappa}
Re \;
V_t^{(n)}(c)
+
\frac {1}{ \kappa}
Re 
\left(
\delta_{nt}(\kappa+\frac cn) - \delta_{nt}(\kappa)
\right)
\quad
\label{Theta-eq}
\end{eqnarray}
for some 
$\epsilon_t^{(n)}$
satisfying 
\[
| \epsilon_t^{(n)} | \le 
Ct + 
C \sqrt{\frac tn}.
\]
\end{lemma}
\begin{proof}
By 
Lemma \ref{fourth}(1)
we have (\ref{Theta-eq}) with 
\beq
\epsilon_t^{(n)}
&:=&
- \frac {
\frac cn
}{
 \kappa(\kappa+\frac cn)
}
\int_0^{nt}
(e^{2i \theta_s(\kappa)}-1)
a(s) F(X_s) ds
\\
&& \qquad + 
\frac {1}{ \kappa}
\Biggl\{
\frac i2 \cdot 
\frac {
\frac cn
}
{
\kappa ( \kappa + \frac cn)
}
\int_0^{nt}
a(s)^2 g_{\kappa+\frac cn}(X_s) F(X_s) ds
\\
&& \qquad 
+\frac {i}{ 2\kappa}
\int_0^{nt}
a(s)^2
\left(
g_{\kappa}(X_s) - g_{\kappa+\frac cn}(X_s)
\right)
F(X_s) ds
\Biggr\}.
\eeq
It then suffices to see
$
| \epsilon_t^{(n)} |
\le 
\frac Cn \int_0^{nt} a(s) ds
\le
Ct +  
C\sqrt{\frac tn}.
$
\QED
\end{proof}
\begin{lemma}
\label{a priori}
\beq
{\bf E}[ | \Psi_t^{(n)}(c ) |]
& \le &
C
\left(
t + \sqrt{ \frac tn } + \frac {1}{\sqrt{n}}
\right), 
\quad
t \ge 0, 
\;
n >0. 
\eeq
\end{lemma}
\begin{proof}
We decompose 
$\delta_t(\kappa)$
as is done in 
(\ref{delta})
to estimate 
$\delta_t(\kappa)$
further.
Let 
\[
\Lambda_t^{(n)}(c )
:=
e^{2i \theta_{nt}(\kappa + \frac cn)}
-
e^{2i \theta_{nt}(\kappa)}
\]
then
\begin{eqnarray*}
&&
\delta^{(1)}_{nt}(\kappa + \frac cn)
-
\delta^{(1)}_{nt}(\kappa)
=
\Lambda_t^{(n)}( c) a(nt) g_{\kappa + \frac cn}(X_{nt})
\nonumber
\\
&&\qquad\qquad\qquad
-
\int_0^{nt} 
a'(s) g_{\kappa + \frac cn}(X_s)
\Lambda_{s/n}^{(n)}(c ) ds
+
O(n^{-1}).
\end{eqnarray*}
$\delta^{(2)}_t$
is also decomposed, 
as in 
(\ref{delta2}), (\ref{decomposition}).
The 
$I^{(\beta)}_{nt}$-term can be written as  
\begin{eqnarray*}
&&
I_{nt}^{(\beta)}(\kappa + \frac cn)
-
I_{nt}^{(\beta)}(\kappa)
=
a(nt)^2 h_{\kappa, n}^{(\beta)}(nt)
\Lambda_t^{(n)}(c )
\nonumber
\\
&& \;
-
\int_0^{nt} (a(s)^2)' f_{\kappa, n}^{(\beta)}(s)
\Lambda_{s/n}^{(n)}(c )
ds
-
\int_0^{nt} a(s)^3 g_{\kappa, n}^{(\beta)}(s)
\Lambda_{s/n}^{(n)}(c )
ds
\end{eqnarray*}
for some bounded functions 
$f_{\kappa, n}^{(\beta)}, g_{\kappa, n}^{(\beta)}, h_{\kappa, n}^{(\beta)}$.
Putting together 
we have 
\beq
\delta_{nt}(\kappa + \frac cn) &-& \delta_{nt}(\kappa)
=
\Lambda_t^{(n)}(c )
\Bigl(
a(nt) g_{\kappa + \frac cn}(X_{nt})
+
a(nt)^2 
h_{\kappa, n}(nt)
\Bigr)
\\
&& + 
\int_0^{nt}
\Lambda_{s/n}^{(n)}(c )
b_{\kappa, n}(s)
ds
 + 
N_{nt}(\kappa + \frac cn) - N_{nt}(\kappa)
+
O(n^{-1})
\eeq
for some bounded functions 
$h_{\kappa, n}, b_{\kappa, n}$
and a martingale
$N_t$.
$b_{\kappa, n}(s)$
is a linear combination of 
$a'(s) g_{\kappa + \frac cn}$,
$( a(s)^2 )' f_{\kappa, n}^{(\beta)}$, 
and 
$a(s)^3 g_{\kappa, n}^{(\beta)}$, 
so that it is integrable : 
$\int_0^{\infty} b_{\kappa, n}(s) ds < \infty$. 
Taking expectations, 
the martingale terms vanish and it follows that 
\beq
{\bf E} \left[
\delta_{nt}(\kappa + \frac cn) - \delta_{nt}(\kappa)
\right]
&=&
{\bf E} \left[
\Lambda_t^{(n)}(c )
\Bigl(
a(nt) g_{\kappa + \frac cn}(X_{nt})
+
a(nt)^2 
h_{\kappa, n}(nt)
\Bigr)
\right]
\\
&& + 
\int_0^{nt}
{\bf E}\left[
\Lambda_{s/n}^{(n)}(c )
b_{\kappa, n}(s)
\right]
ds
+
O(n^{-1}).
\eeq
Therefore 
we can find a non-random function
\[
b(s) = C(a'(s) + (a(s)^2)' + a(s)^3)
\]
for some 
$C>0$ 
such that 
$\int_0^{\infty} b(s) ds < \infty$
and
\beq
&&
\left|
{\bf E} \left[
\delta_{nt}(\kappa + \frac cn) - \delta_{nt}(\kappa)
\right]
\right|
\le 
C a(nt)
{\bf E}[ | \Lambda_t^{(n)}(c ) | ]
+
\int_0^{nt} 
{\bf E}[ | \Lambda_{s/n}^{(n)}(c ) | ] b(s) ds
+
\frac Cn.
\eeq
Here without loss of generality, we may suppose 
$c \ge 0$. 
We use 
$\Psi_t^{(n)}(c ) \ge 0$
for 
$c \ge 0$
and take expectation in 
(\ref{Theta-eq}).
\beq
&&
{\bf E}[ | \Psi_t^{(n)}(c ) |]
=
{\bf E}[ \Psi_t^{(n)}(c ) ]
\\
&=&
2ct + 
{\bf E}[ Re \;\epsilon_t^{(n)} ] 
+
\frac {1}{ \kappa} {\bf E}\left[ 
Re \left(
\delta_{nt}(\kappa + \frac cn) - \delta_{nt}(\kappa) 
\right)
\right]
\\
& \le &
Ct + C\sqrt{\frac tn}
+ C a(nt) 
{\bf E}[ | \Lambda_t^{(n)}(c ) |]
%
+
C \int_0^{nt}
{\bf E} \left[
\left|
\Lambda_{s/n}^{(n)}(c )
\right|
\right] b(s) ds
+
\frac Cn.
\eeq
Let 
\[
\rho_n (t) := C \left(
t + \sqrt{\frac tn} + \frac 1n
\right).
\]
Since 
$
| \Lambda_t^{(n)}(c ) |
\le 
|\Psi_t^{(n)}(c )|
$ 
we have
\beq
&&
{\bf E}\left[ 
\left| 
\Psi_t^{(n)}(c ) 
\right| \right]
%
\le
\rho_n (t) + 
C a(nt) {\bf E}\left[ 
\left| 
\Psi_t^{(n)}(c ) 
\right| \right]
+
C \int_0^{nt}
{\bf E}\left[ 
\left| 
\Psi_{s/n}^{(n)}(c ) 
\right| \right]
b(s) ds.
\eeq
Fix 
$M>0$
arbitrary.
We may suppose 
$nt > M$
since otherwise 
Lemma \ref{a priori}
holds true by 
(\ref{M}).
(\ref{M})
also implies 
$\int_0^M 
{\bf E}\left[ \left|
\Psi_{s/n}^{(n)}
\right| \right] b(s) ds
\le 
\frac Cn$
which gives us
\beq
{\bf E}\left[ 
\left| 
\Psi_t^{(n)}(c ) 
\right| \right]
%
&\le&
\rho_n (t) +
C a(M) {\bf E}\left[ 
\left| 
\Psi_t^{(n)}(c ) 
\right| \right]
\\
&&
\qquad
+
C \int_M^{nt} 
{\bf E}\left[ 
\left| 
\Psi_{s/n}^{(n)}(c ) 
\right| \right] b(s) ds
+
\frac Cn.
\eeq
Take 
$M$
large enough such that 
$C a(M) < 1$
and renew the positive constant 
$C$
in the definition of 
$\rho_n(t)$. 
Then we have
\beq
{\bf E}\left[ 
\left| 
\Psi_t^{(n)}(c ) 
\right| \right]
\le
\rho_n(t)
+
C \int_{M/n}^t
{\bf E}\left[ 
\left| 
\Psi_s^{(n)}(c ) 
\right| \right] 
n b(ns) ds.
\eeq
By
Grownwall's inequality, 
\beq
&&
{\bf E}\left[ 
\left| 
\Psi_t^{(n)}(c ) 
\right| \right]
%
\le
\rho_n (t)
+
C \int_{M/n}^t 
\rho_n (s) n b(ns)
\exp \left(
C \int_s^t n b(nu) du 
\right)
ds.
\eeq
Since 
$b$
is integrable, 
$\exp \left(
C \int_s^t n b(nu) du 
\right)$
is bounded so that 
\begin{equation}
{\bf E}\left[ 
\left| 
\Psi_t^{(n)}(c ) 
\right| \right]
 \le 
\rho_n (t)
+
C \int_{M/n}^t 
\rho_n (s) n b(ns) ds.
\label{estimate-Theta}
\end{equation}
Substituting 
\beq
\int_{M/n}^t \rho_n (s) n b(ns) ds
&=& C
\int_M^{nt}
\left(
\frac sn + \sqrt{\frac {s}{n^2}} + \frac 1n
\right)
b(s) ds
%
\le 
\frac {C}{\sqrt{n}}
\eeq
into
(\ref{estimate-Theta})
yields the conclusion.
\QED
\end{proof}
\begin{lemma}
\label{square-integrable}
For 
$t > 0$, 
we have
\beq
{\bf E}[ \langle V^{(n)}(c ), 
\overline{V^{(n)}(c )} \rangle_t ]
\le
Ct + o(1)
\eeq
as
$n \to \infty$. 
In particular, 
$\sup_n
{\bf E}[ \langle V^{(n)}(c ), 
\overline{V^{(n)}(c )} \rangle_t ]<\infty$. 
\end{lemma}
\begin{proof}
A straightforward computation 
using Lemma \ref{partial integration}(2)
yields
\beq
\langle V^{(n)}(c), \overline{V^{(n)}(c)} \rangle_t
&=&
\int_0^{nt}
a(s)^2
\left|
e^{2i (\theta_s(\kappa+\frac cn) - \theta_s(\kappa))}-1
\right|^2
[g_{\kappa}, \overline{g_{\kappa}}](X_s) ds
+ o(1)
\\
&=&
\langle 
[g_{\kappa}, \overline{g_{\kappa}}]
\rangle
\int_0^{nt}
a(s)^2
\left|
e^{2i (\theta_s(\kappa+\frac cn) - \theta_s(\kappa))}-1
\right|^2
ds  
+ o(1)
\eeq
as 
$n \to \infty$. 
We take expectations and use 
Lemma \ref{a priori}. 
\beq
{\bf E}[ \langle V^{(n)}(c ), 
\overline{V^{(n)}(c )} \rangle_t ]
&=&
C n \int_0^t a(ns)^2 
{\bf E}\left[
| e^{i \Psi_s^{(n)}( c)} - 1 |^2
\right]
ds 
+ o(1)
\\
& \le &
C n  \int_0^t a(ns)^2 
{\bf E} \left[ \left|
\Psi_s^{(n)}( c)
\right| \right] ds 
+ o(1)
\\
&\le&
C\left(
t + \sqrt{ \frac tn } + \frac {\log (nt)}{\sqrt{n}}
\right)
+ o(1).
\eeq
\QED
\end{proof}
\begin{lemma}
\label{a priori2}
For each 
$c > 0$, $T > 0$
fixed we have 
\beq
&&
{\bf E}
\left[ 
\sup_{0 \le t \le T}
\Psi_t^{(n)}(c)
\right]
\\
&\le&
C \left(
T + \sqrt{ \frac Tn }
\right)
+
C T^{\frac 12} + o(1)
+ C
{\bf E} \left[
\max_{0 \le t \le T}
| \delta_{nt}(\kappa + \frac cn) - \delta_{nt}(\kappa) |
\right].
\eeq
as $n \to \infty$. 
\end{lemma}
\begin{proof}
We estimate the third term of 
(\ref{Theta-eq})
by the martingale inequality and 
use Lemma \ref{square-integrable} : 
%
$
{\bf E}\left[
\sup_{0 \le t \le T}
| V_t^{(n)}(\kappa) |
\right]
 \le 
C
{\bf E}\left[
| V_T^{(n)}(\kappa) |^2
\right]^{1/2}
%
\le
C(T + o(1))^{\frac 12}$.
%
\QED
\end{proof}
\begin{lemma}
\label{martingale}
For each 
$0 < t_0 < t_1 < \infty$, 
we can find 
$C = C (t_0, t_1)$
such that for large 
$n$, 
we have 
\beq
{\bf E}
\left[
\left|
V_t^{(n)}(c ) - V_s^{(n)}(c )
\right|^4
\right]
\le
C(t-s)^2
\eeq
for any 
$s,t \in [t_0, t_1]$. 
\end{lemma}
\begin{proof}
By 
martingale inequality, 
\beq
&&
{\bf E}
\left[
\left|
V_t^{(n)}(c )
-
V_s^{(n)}(c )\right|^4
\right]
\le
C
{\bf E}
\left[
\left|
V_t^{(n)}(c )
-
V_s^{(n)}(c )
\right|^2
\right]^2
\\
& \le &
C
{\bf E}
\left[
\int_{ns}^{nt} 
a(u)^2
\left[
G_{\kappa}(u), \overline{G_{\kappa}(u)}
\right]
(X_u) du
\right]^2
\\
& \le &
C
\left(
\int_{ns}^{nt} a(u)^2 du
\right)^2.
\\
&&\mbox{ where }
G_{\kappa}(u)
:=
e^{2i \theta_u(\kappa+\frac cn)}g_{\kappa+\frac cn}
-
e^{2i \theta_u(\kappa)}g_{\kappa}.
\eeq
We can find 
$N=N(t_0)$
such that for 
$n \ge N$
\beq
C
\left(
\int_{ns}^{nt} a(u)^2 du
\right)^2
& \le & C
\log
\left(
1 + \frac {t-s}{t_0}
\right)^2
\le C (t-s)^2.
\eeq
\QED
\end{proof}
\subsection{Tightness of $\Psi$}
\begin{lemma}
\label{tightness}
For any 
$c = (c_1, c_2, \cdots, c_m) \in {\bf R}^m$, 
the sequence of 
${\bf R}^m$-valued process
$\{ \Psi_t^{(n)}(c) \}_{n \ge 1}
=
\{ (
\Psi_t^{(n)}(c_1), \cdots, \Psi_t^{(n)}(c_m)
)\}_{n \ge 1}
$
is tight as a family in 
$C([0,T] \to {\bf R}^m)$.
\end{lemma}
\begin{proof}
It is sufficient to show 
\beq
&(1)&
\lim_{A \to \infty} \sup_n
{\bf P}( | \Psi_t^{(n)} ( c) | \ge A ) = 0
\\
&(2)&\;
\lim_{\delta \downarrow 0}
\limsup_{n \to \infty}
{\bf P}\left(
\sup_{0 \le s, t \le T, \; |t-s| < \delta }
| \Psi_t^{(n)}(c ) - \Psi_s^{(n)}( c) | > \rho 
\right) = 0, 
\quad
T, \rho > 0.
\eeq
(1) follows from 
Lemma \ref{a priori}.
To prove
(2), 
we fix
$M > 0$
arbitrary and decompose
\beq
&&
{\bf P}\left(
\sup_{0 \le s, t \le T, \; |t-s| < \delta }
| \Psi_t^{(n)}(c ) - \Psi_s^{(n)}( c) | > \rho 
\right)
\\
& \le &
{\bf P}\left(
\sup_{0 \le s, t \le M, \; |t-s| < \delta }
| \Psi_t^{(n)}(c ) - \Psi_s^{(n)}( c) | > \rho 
\right)
\\
&& \qquad
+
{\bf P}\left(
\sup_{M \le s, t \le T, \; |t-s| < \delta }
| \Psi_t^{(n)}(c ) - \Psi_s^{(n)}( c) | > \rho 
\right)
=:I+II.
\eeq
Since 
$\Psi_0^{(n)}(c)=0$
we have
\beq
I 
& \le &
{\bf P}\left(
\sup_{t \le M} | \Psi_t^{(n)}(c )| > \frac {\rho}{2}
\right)
+
{\bf P}\left(
\sup_{s \le M} | \Psi_s^{(n)}(c )| > \frac {\rho}{2}
\right)
\eeq
and we use 
Lemma \ref{a priori2}
\beq
&&{\bf P}\left(
\sup_{t \le M} | \Psi_t^{(n)}(c) | > \frac {\rho}{2}
\right)
\le
\frac {2}{\rho}
{\bf E}\left[
\sup_{0 \le t \le M}
| \Psi_t^{(n)}(c) | 
\right]
\\
& \le &
C \left(
M + \sqrt{ \frac Mn } 
\right)
+
C M^{\frac 12} + o(1)
%
+
C
{\bf E} \left[
\max_{0 \le t \le M}
\Bigl| 
\delta_{nt}(\kappa + \frac cn) - \delta_{nt}(\kappa) 
\Bigr|
\right]
\eeq
as
$n \to \infty$. 
By 
Lemma \ref{fourth}(3)
the third term vanishes as 
$n \to \infty$
and it holds that 
$\limsup_{n \to \infty}
I
\le CM^{1/2}$.
Thus following estimate will be sufficient 
\begin{equation}
\lim_{\delta\downarrow 0} \limsup_{n \to \infty} II = 0.
\label{II}
\end{equation}
By
Lemmas \ref{fourth}, \ref{Theta}, 
eq.(\ref{II})
will follow from the following equation
\begin{eqnarray}
&&\lim_{\delta \downarrow 0}
\limsup_{n \to \infty}
\nonumber
{\bf P}\Biggl(
\sup_{M \le t, s \le T, \; |t-s| < \delta}
\Biggl|
V_t^{(n)}(c )
-
V_s^{(n)}(c )
\Biggr|
> \rho
\Biggr)
= 0
\label{Y}
\end{eqnarray}
which, in turn, follows from 
Lemma \ref{martingale}
and Kolmogorov's theorem.
\QED
\end{proof}
%
\subsection{SDE satisfied by 
$\Psi$}
In this subsection 
we show that 
$\Psi^{(n)}$
has a limit 
$\Psi$
which satisfies 
(\ref{SDEbeta}). 
\begin{lemma}
\label{martingale2}
For any
$c_1, \cdots, c_m \in {\bf R}$, 
the solution of the following martingale problem is unique: %
\[
W_t(c_j ) = \Psi_t( c_j) - 2c_j t, 
\quad
j=1, 2, \cdots, m
\]
are martingales whose variational process satisfy
\[
\langle W(c_i ), W(c_j ) \rangle_t
=
D^2
\int_0^t s^{-1}
Re
\left\{
\left(
e^{i \Psi_s( c_i )}-1
\right)
\left(
e^{-i \Psi_s( c_j )}-1
\right)
\right\}
ds.
\]
Moreover 
$\Psi_t(c_j)$
can be characterized by the unique solution to 
the following SDE.
\[
d \Psi_t (c_j)
=
2 c_j dt 
+ D t^{- 1/2}Re
[(e^{i \Psi_t(c_j)}-1) d Z_t ], 
\quad
\Psi_0(c_j) = 0.
\]
\end{lemma}
The proof of 
Lemma \ref{martingale2}
is similar to that of 
Lemma \ref{parameter-continuity}
except that we consider the regularized $p$-th power and use 
Lemma \ref{a priori}. 
\begin{theorem}\mbox{}\\
\label{SDE}
For any
$c_1, \cdots, c_m \in {\bf R}$,
$(\Psi_t^{(n)}(c_1), \cdots, \Psi_t^{(n)}(c_m))
\stackrel{d}{\to}
(\Psi_t(c_1), \cdots, \Psi_t(c_m))$, 
where
$\{ \Psi_t( c_j ) \}$
satisfies (\ref{SDEbeta}). 
\end{theorem}
\begin{proof}
By Lemma \ref{tightness}, 
the sequence 
$\{ (\Psi_t^{(n)}(c_1), \cdots, \Psi_t^{(n)}(c_m)) \}_{n \ge 1}$
has a limit point 
$(\Psi_t(c_1), \cdots, \Psi_t(c_m))$. 
Since Lemmas \ref{fourth}, \ref{Theta} 
imply 
\beq
\Psi_t^{(n)}( c)
&=&
2ct + 
\frac {1}{\kappa}
Re \;
V_t^{(n)}(c )
+
o(1)
\eeq
in probability, we study 
$V_t^{(n)}(c )$.
By a computation using 
Lemma \ref{partial integration}, 
$
\langle V^{(n)}( c), V^{(n)}( c') \rangle_t
\stackrel{n \to \infty}{\to} 0
$
in mean square.
Similarly,
\beq
&&
\langle V^{(n)}( c), \overline{V^{(n)}( c')} \rangle_t
\\
&=&
\langle [g_{\kappa}, \overline{g_{\kappa}}] \rangle
\int_0^{t} n a(nu)^2
\left(
e^{i \Psi_u^{(n)}( c)}-1
\right)
\overline{
\left(
e^{i \Psi_u^{(n)}( c')}-1
\right)
}
du
+ o(1).
\eeq
By Skorohod's theorem, we can suppose 
$
\Psi_t^{(n)} ( c) \to \Psi_t( c)
$
compact uniformly with respect to $t$. 
Hence for 
$0 < s < t$, 
\beq
&&
\langle V^{(n)}( c), \overline{V^{(n)}( c')} \rangle_t
-
\langle V^{(n)}( c), \overline{V^{(n)}( c')} \rangle_s
\\
& \stackrel{n \to \infty}{\to} &
\langle [g_{\kappa}, \overline{g_{\kappa}}] \rangle
\int_{s}^{t} u^{-1}
\left(
e^{i \Psi_u( c)}-1
\right)
\overline{
\left(
e^{i \Psi_u( c')}-1
\right)
}
du.
\eeq
On the other hand by 
Lemma \ref{a priori}
we have
\[
\int_0^t 
{\bf E} \left[
\left|
e^{i \Psi_s(c )} - 1
\right|^2 
\right]
\frac {ds}{s}
\le
C
\int_0^t 
{\bf E}
\left[
\left| 
\Psi_s(c) 
\right|
\right]
\frac {ds}{s}
< \infty
\]
so that 
$V_t(c) = \lim_{n \to \infty}V_t^{(n)}(c )$
is a square integrable continuous martingale 
whose 
variational process
satisfy
\beq
\langle V(c ), V(c' ) \rangle_t
&=& 0
\\
\langle V(c ), \overline{V(c' )} \rangle_t
&=&
\langle [g_{\kappa}, \overline{g_{\kappa}}] \rangle
\int_{0}^{t}
\left(
e^{i \Psi_s( c)}-1
\right)
\overline{
\left(
e^{i \Psi_s( c')}-1
\right)
}
\frac {ds}{s}.
\eeq
Therefore
\[
W_t (c )
=
\Psi_t(c ) - 2ct
=
\frac {1}{\kappa} Re \; V_t(c )
\]
is a square integrable continuous martingale whose variational process is equal to 
\beq
\langle W(c), W(c') \rangle_t
&=&
\frac {
\langle [ g_{\kappa}, \overline{g_{\kappa}}] \rangle
}
{2 \kappa^2}
\int_0^t Re 
\left[
( e^{i \Psi_s(c )}- 1) ( e^{-i \Psi_s(c')}-1)
\frac {ds}{s}
\right].
\eeq
Lemma \ref{martingale2}
yields the conclusion.
\QED
\end{proof}
\begin{lemma}
\label{parameter-continuity}
For 
a.s., 
$\Psi_t(c)$
is continuous on 
$[0, \infty) \times {\bf R}$
and 
is increasing with respect to 
$c$. 
\end{lemma}
\begin{proof}
We shall show the following inequality :  
for
$p > 1$
sufficiently close to $1$, 
\begin{equation}
{\bf E}[ |\Psi_t(c_1) - \Psi_t(c_2)|^p ] 
\le
\frac {2^p(c_1 - c_2)^p}
{1 - \frac 12(p-1) D^2} 
t^p.
\label{Kolmogorov}
\end{equation}
Hence by Kolmogorov's theorem, 
for any fixed 
$t > 0$, 
$\Psi_t(c)$
has a continuous version with respect to 
$c\in {\bf R}$ a.s..
We first note that 
$\Psi_t(c)$
satisfies 
\beq
d \Psi_t(c)
&=&
2c dt 
+
\frac {D}{2\sqrt{t}}
\left\{
(e^{i\Psi_t} + e^{-i \Psi_t} - 2) d B^1_t
+
i (e^{i\Psi_t} - e^{-i \Psi_t}) d B^2_t
\right\}.
\eeq
Here we note that if 
$c_1 > c_2$
then 
$\Psi_t(c_1) > \Psi_t(c_2)$
by the comparison theorem of SDE which proves the desired monotonicity of 
$\Psi_t(c)$. 
We set 
\beq
\Gamma_t 
&:=&
\Psi_t(c_1) - \Psi_t(c_2), 
\quad
%
\Xi_t
:=
e^{i\Psi_t(c_1)}  - e^{i\Psi_t(c_2)}.
\eeq
For 
$c_1 > c_2$,
we see 
\beq
d \Gamma_t
&=&
2(c_1 - c_2) dt
+
\frac {D}{2 \sqrt{t}}
\left\{
(\Xi_t + \overline{\Xi}_t) d B^1_t
+
i( \Xi_t - \overline{\Xi}_t ) d B^2_t
\right\}.
\eeq
Hence
\beq
(d\Gamma_t)^2
&=&
\frac {D^2}{4t}
\left\{
(\Xi_t + \overline{\Xi}_t)^2 
- 
(\Xi_t - \overline{\Xi}_t)^2
\right\}dt
=
\frac {D^2}{t} | \Xi_t |^2 dt.
\eeq
Then for 
$p>1$
\beq
d \Gamma_t^p
&=&
p \Gamma_t^{p-1} d\Gamma_t
+
\frac {p(p-1)}{2}
\Gamma_t^{p-2}
(d \Gamma_t)^2
\\
&=&
2(c_1 - c_2) p \Gamma_t^{p-1} dt
+
\frac {p(p-1)}{2} 
\Gamma_t^{p-2} \frac {D^2}{t} | \Xi_t |^2 dt
\\
&&+
p\Gamma_t^{p-1}
\frac {D}{2 \sqrt{t}}
\left\{
(\Xi_t + \overline{\Xi_t}) dB^1_t
+
i(\Xi_t - \overline{\Xi_t}) dB^2_t
\right\}.
\eeq
Taking expectation yields
\begin{equation}
{\bf E}[ \Gamma_t^p ]
=
2(c_1 - c_2)p 
\int_0^t 
{\bf E}[ \Gamma_s^{p-1} ] ds
+
\frac {p(p-1)}{2}
D^2 
\int_0^t 
{\bf E}[ \Gamma_s^{p-2} | \Xi_s |^2 ]
\frac {ds}{s}.
\label{p-th power}
\end{equation}
We have
%
$| \Xi_t |^2 \le
C\Gamma_t^{\gamma}$, 
$0 < \gamma < 2$
%
for some positive constant
$C$ 
and some 
$0 < \gamma < 2$.
Hence
\beq
\int_0^t 
{\bf E}[ \Gamma_s^{p-2} | \Xi_s |^2 ]
\frac {ds}{s}
& \le &
C
\int_0^t 
{\bf E}[ \Gamma_s^{p-2 + \gamma}]
\frac {ds}{s}.
\eeq
We use  
${\bf E}[|X|^r] \le {\bf E}[|X|]^r$
for 
$r \le 1$
and the fact that
${\bf E}[ \Psi_t(c) ] = 2c t$.  
Assuming 
$p-1 \le 1$
and
$0 < p-2 + \gamma \le 1$
yields
\beq
{\bf E}[\Gamma_t^p]
& \le &
2(c_1 - c_2) p 
\int_0^t {\bf E}[\Gamma_s]^{p-1} ds
+
C \int_0^t 
{\bf E}[ \Gamma_s ]^{p-2+\gamma} 
\frac {ds}{s}
\\
&=&
2^p (c_1 - c_2)^p t^p
+
C(c_1 - c_2)^{p-2+\gamma} 
t^{p-2 + \gamma}
\eeq
so that for 
$0 \le t \le T$
we have
\beq
f(t) := {\bf E}[ \Gamma_t^p ]
\le C_T 
t^{p-2 + \gamma}
\eeq
and hence 
\beq
h(t) 
:=
\int_0^t 
\frac {f(s)}{s} ds
\le C_T 
t^{p-2 + \gamma}.
\eeq
Thus for any 
$p>1$
sufficiently close to 
$1$, 
we take 
$\gamma$
satisfying 
$1+(p-1)(\frac p2 D^2-1) < \gamma \le 3-p$
so that 
\begin{equation}
h(t) \le C t^{\frac p2(p-1) D^2 + \delta}
\label{a priori-eq}
\end{equation}
for some 
$\delta > 0$.
On the other hand 
by using 
$
| \Xi_s |^2 \le  \Gamma_s^2
$
in 
(\ref{p-th power}) 
we have
\beq
{\bf E}[ \Gamma_t^p ]
& \le &
2(c_1 - c_2) p 
\int_0^t 
2^{p-1}(c_1 - c_2)^{p-1} s^{p-1} ds
 + 
\frac p2 (p-1) D^2 
\int_0^t 
{\bf E}[ \Gamma_s^p ] \frac {ds}{s}
\\
&=&
2^p (c_1 - c_2)^p t^p 
+
\frac p2 (p-1) D^2 
\int_0^t {\bf E}[ \Gamma_s^p ] \frac {ds}{s}.
\eeq
Hence if 
$\frac 12 (p-1) D^2 < 1$, 
(\ref{a priori-eq}) and a 
Grownwall type argument
give the desired inequality
(\ref{Kolmogorov}). 

Having established 
the continuity of 
$\Psi_{t_0}(c)$
with respect to 
$c$, 
the joint continuity of 
$\Psi_t(c)$
on
$[t_0, \infty) \times {\bf R}$
is valid due to the absence of singularity in this time domain. 
The continuity of 
$\Psi_t(c)$
at 
$t=0$
follows from the monotonicity of 
$\Psi_t(c)$
with respect to 
$c$. 
\QED
\end{proof}
%
%
\begin{remark}
$\{ \Psi_t(c ) \}_{t \ge 0, c \in {\bf R}}$
satisfies the following properties: 
\\
(1)
The process has invariance
\[
\{ \Psi_t( c) \}_{t \ge 0, c \in {\bf R} }
\stackrel{law}{=}
\{ \Psi_t( c+ c_0) - \Psi_t( c_0) \}_{t \ge 0, c \in {\bf R} }
\]
for any 
$c_0 \in {\bf R}$.
\\
(2)
For each fixed 
$c$
there exists a 1-D Brownian motion 
$\{ B_t(c ) \}$ 
such that 
\beq
\frac {\partial\Psi_t}{\partial c}
=
2 \int_{0}^t 
\exp \left(
\int_s^t \frac {D}{\sqrt{u}} d B_u
-
\int_s^t \frac {D^2}{2u} du
\right)
ds
\eeq
where 
$\{ B_t(c )\}$
is a family of martingales satisfying
\[
\langle B_{\cdot}(c ), B_{\cdot}(c') \rangle_t
=
\int_0^t 
\cos \left(
\Psi_s(c ) - \Psi_s(c') 
\right) ds.
\]
\end{remark}
\section{Convergence of 
$\theta_t(\kappa)$ mod $\pi$}
\begin{proposition}
\label{uniform}
As 
$t \to \infty$
$(2 \theta_t(\kappa))_{2 \pi {\bf Z}}$
converges to the uniform distribution on 
$[0, 2 \pi)$. 
\end{proposition}
\begin{proof}
Letting 
$
\xi_t(\kappa) := e^{2 m i \tilde{\theta}_t(\kappa)}, 
m \in {\bf Z}, 
$
it suffices to show 
$
{\bf E}[ \xi_t (\kappa) ] \stackrel{t \to \infty}{\to} 0, 
m \ne 0.
$
We omit the 
$\kappa$-dependence of 
$\theta_t$.
By
(\ref{diff-eq})
we decompose
\beq
\xi_t
&=&
1 + 
\frac {mi}{2 \kappa}
\int_0^t
e^{2i \kappa s + 2(m+1) i\tilde{\theta}_s}
a(s) F(X_s)ds
\\
&& + 
\frac {mi}{2 \kappa}
\int_0^t
e^{-2i \kappa s + 2(m-1) i\tilde{\theta}_s}
a(s) F(X_s)ds
%
- \frac {mi}{\kappa}
\int_0^t
e^{2m i\tilde{\theta}_s}
a(s) F(X_s)ds
\\
&=:& 1 + I + II + III.
\eeq
We use 
Lemma \ref{partial integration}(1)
and decompose 
$I$
further : 
\begin{eqnarray}
I
&=&
\frac {mi}{2 \kappa}
\Biggl(
-\frac {2i(m+1)}{4 \kappa}
\int_0^t a(s)^2 
e^{2mi \tilde{\theta}_s}
F(X_s) g_{\kappa}(X_s) ds
\nonumber
\\
&& + 
\int_0^t a(s) e^{2i \kappa s} 
e^{2i(m+1) \tilde{\theta}_s} 
dM_s(\kappa)
+ \delta_{1,1}(t)
\Biggr).
\label{1}
\end{eqnarray}
where
\beq
&&
\delta_{1,1}(t)
\\
&:=&
\left[
a(s) e^{2i (m+1) \tilde{\theta}_s}
e^{2i \kappa s}
g_{\kappa}(X_s) 
\right]_0^t
%
- 
\int_0^t 
a'(s) e^{2i (m+1) \tilde{\theta}_s}
e^{2i \kappa s}
g_{\kappa}(X_s) ds
\\
&& - 
\frac {2i (m+1)}{2 \kappa}
\int_0^t a(s)^2 
\left(
\frac {
e^{2(m+2)i \tilde{\theta}_s}e^{4i \kappa s}
}
{2}
-
e^{2(m+1) i \tilde{\theta}_s} 
e^{2i \kappa s}
\right)
F(X_s) g_{\kappa}(X_s) ds.
\eeq
We further compute 
the third term of 
$\delta_{1,1}$
by 
Lemma \ref{partial integration}(1)
and see that 
$\delta_{1,1}(t)$
has a limit as 
$t \to \infty$.
Taking expectation, 
martingale term vanishes and we have 
\begin{equation}
{\bf E}[ \delta_{1,1}(t) ] - {\bf E}[ \delta_{1,1}(\infty) ] 
=
O(a(t)), 
\quad
t \to \infty.
\label{infty}
\end{equation}
By 
Lemma \ref{partial integration}(2), 
the first term of 
(\ref{1})
satisfies
\beq
&&
\int_0^t a(s)^2 
e^{2mi \tilde{\theta}_s}
F(X_s) g_{\kappa}(X_s) ds
=
\langle F g_{\kappa} \rangle
\int_0^t a(s)^2 e^{2mi \tilde{\theta}_s}ds 
+ 
\delta_{1,2}(t)
\eeq
where
$\delta_{1,2}(t)$
has a limit as 
$t \to \infty$
and satisfies the same estimate as 
(\ref{infty}).
We substitute it into 
(\ref{1})
and let 
$\delta_1 = \delta_{1,1} + \delta_{1,2}$.
Then 
\begin{eqnarray}
I
&=&
\frac {mi}{2 \kappa}
\Biggl(
-\frac {2i(m+1)}{4 \kappa}
\langle F g_{\kappa} \rangle
\int_0^t a(s)^2 
e^{2mi \tilde{\theta}_s}
ds
\nonumber
\\
&& +
\int_0^t a(s) e^{2i \kappa s} 
e^{2i(m+1) \tilde{\theta}_s} 
dM_s(\kappa)
+ \delta_{1}(t)
\Biggr).
\label{1-1}
\end{eqnarray}
We compute II, III 
in a similar manner and consequently  
%
%
\begin{eqnarray}
\xi_t
&=&
1 - \frac {mi}{\kappa} \langle F \rangle
\int_0^t a(s) e^{2mi \tilde{\theta}_s}ds
+
\langle G_m \rangle
\int_0^t a(s)^2 e^{2mi\tilde{\theta}_s}
ds
+ N_t
+ \delta(t)
\qquad
\label{xi}
\end{eqnarray}
where
\beq
G_m&=&
\left(
\frac {m(m+1)}{4 \kappa^2} g_{\kappa}
+
\frac {m(m-1)}{4 \kappa^2} g_{- \kappa}
+
\frac {m^2}{\kappa^2} g
\right)F
\\
N_t &=&
 \frac {mi}{2 \kappa}
\int_0^t a(s) 
e^{2i \kappa s} e^{2i(m+1)\tilde{\theta}_s}
dM_s(\kappa)
%
+ \frac {mi}{2 \kappa}
\int_0^t a(s) 
e^{-2i \kappa s} e^{2i(m-1) \tilde{\theta}_s}
dM_s(-\kappa)
\\
&&
- \frac {mi}{\kappa}
\int_0^t a(s) e^{2mi \tilde{\theta}_s}
dM_s
\eeq
where
$\delta(\infty) = \lim_{t \to \infty} \delta(t)$
exists a.s. and 
\[
{\bf E}[ \delta(t) ]
-
 {\bf E}[ \delta(\infty) ] 
=
O(a(t)), 
\quad
t \to \infty.
\]
Let 
$\sigma_F(d \lambda)$
be the spectral measure of 
$L$
with respect to 
$F$. 
Then by noting 
\beq
Re 
\langle F g_{\kappa} \rangle
=
Re \langle F g_{- \kappa} \rangle
=
\int_{-\infty}^0
\frac {\lambda\sigma_F(d \lambda)}{\lambda^2 + 4\kappa^2}
< 0, 
\quad
Re \langle F g \rangle
=
\int_{-\infty}^0
\frac {\sigma_F(d \lambda)}{\lambda}
< 0
\eeq
we have
$
- \gamma := Re \langle G_m \rangle < 0.
$
Set 
\[
\rho(t) := {\bf E}[ \xi_t ], 
\quad
b(t) :=
- \frac {mi}{\kappa} \langle F \rangle a(t)
+
\langle G_m \rangle a(t)^2.
\]
Then 
(\ref{xi})
turns to 
\[
\rho(t) 
=
1 + \int_0^t b(s) \rho(s) ds 
+
{\bf E}[ \delta (t) ]
\]
and hence
\beq
\rho(t)
&=&
\exp \left(
\int_0^t b(u) du
\right)
+
{\bf E}[ \delta(t) ] 
+ 
\int_0^t 
{\bf E}[ \delta (s) ] 
b(s)
\exp \left(
\int_s^t b(u) du
\right)
ds
\\
&=&
\exp \left(
\int_0^t b(u) du
\right)
+
{\bf E}[ \delta(t) ] 
+ 
{\bf E}[ \delta(\infty) ]
\int_0^t 
b(s)
\exp \left(
\int_s^t b(u) du
\right)
ds
\\
&& +
\int_0^t 
\left(
{\bf E}[ \delta(s) ] - {\bf E}[ \delta(\infty) ] 
\right)
b(s)
\exp \left(
\int_s^t b(u) du
\right)
ds
\\
&=:& I + II + III + IV.
\eeq
Noting 
$
Re \; b(t) = Re \;\langle G_m \rangle a(t)^2
=
- \gamma a(t)^2
$, 
we compute 
$I, III$ 
\beq
| I | 
& \le &
\exp \left(
\int_0^t Re\, b(s) ds
\right)
\le
C 
\exp \left(
- \gamma \int_1^{t} \frac 1s ds
\right)
\stackrel{t \to \infty}{\to} 0
\\
III
&=&
{\bf E}[ \delta (\infty) ] 
\left(
-1 + 
\exp \left(
\int_0^t b(u) du
\right)
\right)
\stackrel{t \to \infty}{\to} 
- {\bf E}[ \delta(\infty) ].
\eeq
We further decompose 
$IV$ : 
\beq
| IV |
&=&
\left|
\int_0^t 
\left(
{\bf E}[ \delta(s) ] - {\bf E}[ \delta(\infty) ] 
\right)
b(s)
\exp \left(
\int_s^t b(u) du
\right)
ds
\right|
\\
& \le &
C
\left(
\int_0^M + \int_M^{t} 
\right)
a(s)|b(s)|
\exp \left(
Re\, \int_s^t b(u) du
\right)
ds
\\
&=:&IV_1 + IV_2.
\eeq
It is easy to see that 
$IV_1 \stackrel{t \to \infty}{\to} 0$. 
For 
$IV_2$
we use 
$\langle F \rangle = 0$
and compute, for large 
$M$, 
\beq
| IV_2 |
& \le &
C
\int_M^t 
a(s)^3 
\exp \left(
\int_s^t Re \, b(u) du
\right)ds
\\
& \le &
C \int_M^t 
s^{-3/2}
\left( \frac ts \right)^{- \gamma} ds
%
=
\left\{
\begin{array}{@{\,}ll}
C t^{-\gamma} \log \frac {t}{M}
& 
(\gamma=\frac 12)
\\
C t^{-\gamma}\,
\frac {
t^{\gamma- \frac 12} - M^{\gamma - \frac 12}
}
{
\gamma - \frac 12
}
&
( \gamma \ne \frac 12 )
\\
\end{array}
\right.
\quad
\stackrel{t \to \infty}{\to} 0.
\eeq
\QED
\end{proof}
%

\section{Limiting behavior of $\tilde{\theta}_t$}
To study the limiting behavior of 
$(2\tilde{\theta}_t)_{2 \pi {\bf Z}}$
we set 
\[
\tilde{\xi}_t := e^{2i \tilde{\theta}_t(\kappa)}.
\]
\subsection{Estimate of integral equation}
As in the proof of 
Proposition \ref{uniform}, 
we can show the following lemma. 
\begin{lemma}
\label{integral-eq}
Let 
$0 < t_0 < t$.
Then we have
\beq
\tilde{\xi}_t
&=&
\tilde{\xi}_{t_0}
+
\frac {1}{2 \kappa^2}
\langle F \cdot (g_{\kappa}+2g) \rangle
\int_{t_0}^t a(s)^2 e^{2i \tilde{\theta}_s} ds
- \frac {i}{\kappa}
\langle F \rangle
\int_{t_0}^t 
a(s) e^{2i \tilde{\theta}_s} ds
\\
&& 
\qquad + 
\frac {i}{2 \kappa}
\left( Y_t + \widetilde{Y}_t - 2\widehat{Y}_t \right)
+ O(a(t_0)), 
\quad
t_0 \to \infty.
\\
\mbox{ where }\quad
Y_t &:=&
\int_{t_0}^t a(s)
e^{2i\kappa s + 4i \tilde{\theta}_s}
dM_s(\kappa),
%
\quad
\widetilde{Y}_t 
:= 
\int_{t_0}^t a(s) e^{-2i \kappa s} 
dM_s(-\kappa)
\\
\widehat{Y}_t
&:=&
\int_{t_0}^t
a(s) e^{2i \tilde{\theta}_s} 
dM_s.
\eeq
The variational process of 
$Y, \tilde{Y}$, 
and 
$\hat{Y}$
satisfy, as 
$t_0 \to \infty$, 
\beq
\langle Y, Y \rangle_t
&=& O(a(t_0)), 
\;
\langle Y, \overline{Y} \rangle_t
=
\langle [g_{\kappa}, \overline{g}_{\kappa} ]
\rangle
\int_{t_0}^t a(s)^2 ds + O(a(t_0))
\\
\langle \widetilde{Y}, \widetilde{Y} \rangle_t
&=&
O(a(t_0)), 
\;
\langle \widetilde{Y}, \overline{\widetilde{Y}} \rangle_t
=
\langle [g_{\kappa}, \overline{g}_{\kappa}] \rangle
\int_{t_0}^t a(s)^2 ds + O(a(t_0))
\\
\langle \widehat{Y}, \widehat{Y} \rangle_t
&=&
\langle [g, g] \rangle
\int_{t_0}^t a(s)^2 e^{4i \tilde{\theta}_s} ds
+ O(a(t_0)), 
\\
\langle \widehat{Y}, \overline{\widehat{Y}} \rangle_t
&=&
\langle [g, g] \rangle
\int_{t_0}^t a(s)^2 ds + O(a(t_0)).
\eeq
\end{lemma}
%
%
\subsection{Tightness of $\eta$}
Let 
\beq
\eta_t^{(n)} &:=& \tilde{\xi}_{nt} = e^{2i \tilde{\theta}_{nt}(\kappa)},
\quad
{\bf U} := 
\{ z \in {\bf C} \, | \, |z| = 1 \}.
\eeq
\begin{lemma}
\label{tightness-eta}
$\{ \eta_t^{(n)} \}_{n \ge 1}$
is tight as a family in 
$C((0, \infty) \to {\bf U})$. 
\end{lemma}
\begin{proof}
It suffices to show, for any 
$t_0 > 0, \rho > 0$, 
\beq
\lim_{\delta \downarrow 0}
\limsup_{n \to \infty}
{\bf P}
\left(
\sup_{t_0 < s < t, t-s < \delta}
| \eta_t^{(n)} - \eta_s^{(n)} | > \rho 
\right)
= 0.
\eeq
Noting 
$\langle F \rangle = 0$, 
Lemma \ref{integral-eq}
implies
\begin{eqnarray}
\tilde{\xi}_{nt} - \tilde{\xi}_{ns}
&=&
\frac {1}{2 \kappa^2}
\langle F \cdot (g_{\kappa}+2g) \rangle
\int_{ns}^{nt} a(u)^2 
e^{2i \tilde{\theta}_u} du
%
+\frac {i}{2 \kappa}
W^{(n)}_{t,s}
+ o(1), 
\qquad
\label{integral}
\\
\mbox{ where }
\quad
W^{(n)}_{t,s} &:=&
\left(
Y_{nt} + \tilde{Y}_{nt} - 2 \hat{Y}_{nt}
\right)
-
\left(
Y_{ns} + \tilde{Y}_{ns} - 2 \hat{Y}_{ns}
\right), 
\nonumber
\end{eqnarray}
as $n \to \infty$. 
We note that 
$W_{t,s}^{(n)}$
satisfies the estimate in 
Lemma \ref{martingale}
and the rest of the argument is the same as that in 
Lemma \ref{tightness}.
\QED
\end{proof}
%
\subsection{Identification of $\eta_t$}
Let 
$\eta_t$
be a limit point of 
$\eta_t^{(n)}$
which is uniformly distributed on 
${\bf U}$
for each fixed 
$t > 0$ 
by Lemma \ref{uniform}. 
In this subsection
we show that the distribution of the process 
$\eta_t$
is uniquely determined. 
\begin{lemma}
\label{independence of eta}
(1)
For any 
$0 <  t_0 < t$, 
\beq 
\lim_{n \to \infty}
{\bf E}[ 
e^{2mi (\tilde{\theta}_{nt} - \tilde{\theta}_{n t_0})}
| {\cal F}_{n t_0}]
=
\left(
\frac {t}{t_0}
\right)^{\langle G_m \rangle}
\eeq
where
${\cal F}_t$
is the 
$\sigma$-algebra
generated by 
$\{ X_s \}_{0 \le s \le t}$. \\
(2)
For any 
$0 < t_0 < t_1 < \cdots < t_k$, 
the family of random variables
$\{ \eta_{t_0}, \eta_{t_1}/\eta_{t_0}, 
\cdots, 
\eta_{t_k}/\eta_{t_{k-1}} \}$
are independent. 
\end{lemma}
\begin{proof}
(1)
Let
$m, m' \in {\bf Z}$.
By a argument similar to deduce
(7.4), we have
\beq
e^{
2mi (\tilde{\theta}_{nt} - \tilde{\theta}_{n t_0})
}
&=&
1 + \langle G_m \rangle
\int_{t_0}^t
n a(nu)^2
e^{2mi (\tilde{\theta}_{nu} - \tilde{\theta}_{n t_0})}
du
\\
&& \qquad
+
N_{nt, nt_0} e^{-2mi \tilde{\theta}_{nt_0}}
+
\delta_n(t) e^{-2mi \tilde{\theta}_{nt_0}}
\eeq
where
\beq
G_m&=&
\left(
\frac {m(m+1)}{4 \kappa^2} g_{\kappa}
+
\frac {m(m-1)}{4 \kappa^2} g_{- \kappa}
+
\frac {m^2}{\kappa^2} g
\right)F
\\
N_{nt, nt_0} &=&
 \frac {mi}{2 \kappa}
\int_{nt_0}^{nt} a(s) 
e^{2i \kappa s} e^{2i(m+1)\tilde{\theta}_s}
dM_s(\kappa)
\\
&&  
+ \frac {mi}{2 \kappa}
\int_{nt_0}^{nt} a(s) 
e^{-2i \kappa s} e^{2i(m-1) \tilde{\theta}_s}
dM_s(-\kappa)
%
- \frac {mi}{\kappa}
\int_{nt_0}^{nt} a(s) e^{2mi \tilde{\theta}_s}
dM_s
\\
{\bf E}[\delta_n(t)| &{\cal F}_{n t_0}&]
\stackrel{n \to \infty}{\to}
0, 
\quad
a.s..
\eeq
Taking a conditional expectation and 
letting 
\[
\rho_n (t) := 
{\bf E}[ 
e^{2mi (\tilde{\theta}_{nt} - \tilde{\theta}_{n t_0})}
| {\cal F}_{n t_0}]
\]
we have
\beq
\rho_n(t)
&=&
1 + \langle G_m \rangle
\int_{t_0}^t
n a(nu)^2
\rho_n (u) du
+
{\bf E}[\delta_n(t)| {\cal F}_{n t_0}] e^{-2mi \tilde{\theta}_{nt_0}}.
\eeq
Therefore
\beq
\rho_n(t)
&=&
\exp \left(
\langle G_m \rangle
\int_{t_0}^t
n a(nu)^2 du
\right)
+
{\bf E}[\delta_n(t)| {\cal F}_{n t_0}] e^{-2mi \tilde{\theta}_{nt_0}}
\\
+
&\int_{t_0}^t&
{\bf E}[\delta_n(s)| {\cal F}_{n t_0}] e^{-2mi \tilde{\theta}_{nt_0}}
\langle G_m \rangle
n a(nu)^2
\exp \left(
\langle G_m \rangle
\int_s^t
n a(nu)^2 du
\right)
ds
\\
& \stackrel{n \to \infty}{\to}&
\exp \left(
\langle G_m \rangle
\int_{t_0}^t
\frac {du}{u}
\right)
=
\left(
\frac {t}{t_0}
\right)^{\langle G_m \rangle}.
\eeq
(2)
The required independence easily follows from (1) and the fact that 
$e^{2i \tilde{\theta}_{nt}}$
converges to the uniform distribution on 
${\bf U}$
as 
$n \to \infty$. 
In fact, for 
$k=1$,
\beq
&&
{\bf E}\left[
{\bf E}[
e^{
2mi (\tilde{\theta}_{nt} - \tilde{\theta}_{n t_0})
}
|
{\cal F}_{n t_0}
]
e^{2m'i \tilde{\theta}_{n t_0}}
\right]
\\
&=&
{\bf E}\left[
\left(
{\bf E}[
e^{
2mi (\tilde{\theta}_{nt} - \tilde{\theta}_{n t_0})
}
|
{\cal F}_{n t_0}
]
-
\left(
\frac {t}{t_0}
\right)^{\langle G_m \rangle}
\right)
e^{2m'i \tilde{\theta}_{n t_0}}
\right]
\\
&& \qquad
+
{\bf E} \left[
\left(
\frac {t}{t_0}
\right)^{\langle G_m \rangle}
e^{2m'i \tilde{\theta}_{n t_0}}
\right]
\\
&\to&
\left\{
\begin{array}{@{\,}ll}
0 & (m' \ne 0) \\
\left( \frac {t}{t_0} \right)^{\langle G_m \rangle}
& 
(m'=0) \\
\end{array}
\right.
\eeq
For 
$k \ge 2$, 
the proof is similar. 
\QED
\end{proof}
\begin{lemma}
\label{Lemma-SDE-eta}
For each fixed 
$t_0 > 0$, 
$\eta_t$
satisfies the following SDE on 
$t \ge t_0$: 
\begin{eqnarray}
d \eta_t
&=&
C_1 \frac {\eta_t}{t} dt
+
C_2 \frac {\eta_t}{\sqrt{t}} dB_t, 
\label{SDE-eta}
\\
\mbox{where}
\quad
C_1 &:=&
\frac {
\langle ( g_{\kappa}+2g )F  \rangle
}
{ 2 \kappa^2 }, 
\quad
C_2 :=
\frac {i}{2 \kappa}
\sqrt{
\langle 2[g_{\kappa}, \overline{g_{\kappa}}] + 4[g,g] \rangle}.
\nonumber
\end{eqnarray}
\end{lemma}
\begin{proof}
Letting 
$s=t_0 > 0$
in 
(\ref{integral})
yields, as 
$n \to \infty$, 
\beq
\int_{nt_0}^{nt} a(u)^2 
e^{2i \tilde{\theta}_u} du
&\to& 
\int_{t_0}^{t} 
\frac {\eta_u}{u}
du
\\
\langle W^{(n)}_{\cdot, t_0}, W^{(n)}_{\cdot, t_0} \rangle_t
&\to&
\langle 
2[g_{\kappa}, \overline{g_{\kappa}}] 
+ 4[ g,g] 
\rangle
\int_{t_0}^t \frac{\eta_u^2}{u} du
\\
\langle W^{(n)}_{\cdot, t_0}, \overline{W^{(n)}_{\cdot, t_0}} \rangle_t
&\to&
\langle 2 [g_{\kappa}, \overline{g_{\kappa}}]
+ 4 [g, g] \rangle
\int_{t_0}^t \frac {du}{u}.
\eeq
We then proceed as in the proof of 
Theorem \ref{SDE}.
\QED
\end{proof}
\begin{remark}
$Z_t, B_t$
which appear in SDE's
(\ref{SDEbeta}), (\ref{SDE-eta}) of 
$\Psi$, $\eta$
are not independent.
In fact,
$W_t = \lim_{n \to \infty} W_t^{(n)}$, 
$V_t = \lim_{n \to \infty} V_t^{(n)}$
satisfy
\beq
dW_t
&=&
\sqrt{
2 \langle [g_{\kappa}, \overline{g_{\kappa}}] \rangle
+
4[g, g]
}
\frac {\eta_t}{\sqrt{t}} d B_t
\\
d V_t
&=&
\sqrt{ [g_{\kappa}, \overline{g_{\kappa}}] \rangle}
\left(
e^{i\Psi_t(c)}-1
\right)
\frac {d Z_t}{\sqrt{t}}
\\
d \langle W, V \rangle
&=&
\langle [g_{\kappa}, \overline{g_{\kappa}}] \rangle
\left(
e^{i\Psi_t(c)} - 1
\right)
\eta_t \frac {dt}{t}
\\
d \langle W, \overline{V} \rangle
&=&
\langle [g_{\kappa}, \overline{g_{\kappa}}] \rangle
\left(
e^{-i\Psi_t(c)} - 1
\right)
\eta_t \frac {dt}{t}
\eeq
which imply
\beq
dZ dB
&=&
\sqrt{
\frac{
\langle [g_{\kappa}, \overline{g_{\kappa}}] \rangle
}
{
2 \langle [g_{\kappa}, \overline{g_{\kappa}}] 
+ 4 [g,g]
\rangle
}
}
dt.
\eeq
\end{remark}
Here we note the following fact. 
By the time change
$u = \log t$, 
$\zeta_u := \log \eta_{e^u}$
satisfies the following SDE which is stationary in time. 
\begin{eqnarray}
d \zeta_u &=& i C_3 du + i C_4 d \tilde{B}_u
\label{time-change}
\\
\mbox{ where }
\quad
C_3 &:=&
-\frac {1}{\kappa}
\langle | g_{\kappa} |^2 \rangle
\in {\bf R}, 
\quad
C_4
:=
\frac{1}{2 \kappa}
\sqrt{
2 [ g_{\kappa}, \overline{g}_{\kappa} ]
+
4 [ g, g ]
}
\in {\bf R}
\nonumber
\end{eqnarray}
To summarize, 
the following facts have been proved. 

\noindent
(i)
For any 
$t > 0$, 
$\eta_t$
is uniformly distributed(Lemma \ref{uniform}).\\
(ii)
For any 
$0 < t_0 < t_1 < t_2 < \cdots < t_n$, 
random variables

$\{ 
\eta_{t_0}, \eta_{t_1}/\eta_{t_0}, 
\cdots
\eta_{t_n} / \eta_{t_{n-1}}
\}$
are independent(Lemma \ref{independence of eta}).\\
(iii)
For any 
$t_0 > 0$, 
$x_t = \eta_t / \eta_{t_0}$
satisfies an SDE on 
$t \ge t_0$
(Lemma \ref{Lemma-SDE-eta}) : 
\[
dx_t = C_1 \frac{x_t}{t} dt
+
C_2 \frac {x_t}{\sqrt{t}} d B_t, 
\quad
x_{t_0} = 1.
\]
These facts
determines (in distribution) the process
$\eta_t$
uniquely. 
In fact, for any 
$0 < t_0 < t_1 < \cdots < t_n$, 
the distribution of 
$\{ \eta_{t_0}, \eta_{t_1}, \cdots, \eta_{t_n} \}$
can be computed from that of 
$\{ 
\eta_{t_0}, \eta_{t_1}/ \eta_{t_0}, 
\cdots, 
\eta_{t_n} / \eta_{t_{n-1}} 
\}$
and the latter distribution can be determined uniquely from (ii) and (iii). 
Therefore 
the distribution of 
$\{ \eta_t \}$
is characterized by the constants 
$C_1, C_2$. 
More concretely, 
if we prepare 
1D Brownian motion 
$\{ B_t \}_{t \in {\bf R}}$
with 
$B_0 = 0$
and independent random variable
$X \in {\bf C}$
with uniform distribution on 
${\bf U}$, 
a process
$
X \exp \left[
i (C_3 u + C_4 B_u) 
\right]
$
has the same distribution as 
$\{ \eta_{e^u} \}$
by (\ref{SDE-eta}), (\ref{time-change}). 
\section{Convergence of the joint distribution}
We finish the proof of Theorem \ref{sc-limit}.
\subsection{Behavior of the joint distribution}
\begin{proposition}
\label{joint-convergence}
For any
$c_1, \cdots, c_m \in {\bf R}$, 
$t > 0$, 
\begin{equation}
(\Psi_t^{(n)}(c_1), \cdots, \Psi_t^{(n)}(c_m), 
(\theta_{nt}(\kappa))_{2 \pi {\bf Z}})
\stackrel{d}{\to}
(\Psi_t(c_1), \cdots, \Psi_t(c_m), 
\phi_t), 
\label{joint-statement}
\end{equation}
as
$n \to \infty$, 
where
$(\Psi_t(c_1), \cdots, \Psi_t(c_m))$
and 
$\phi_t$
are independent and 
$\phi_t$
is uniformly distributed on 
$[0, 2\pi)$.
\end{proposition}
\begin{proof}
For simplicity, 
we use the following notation. 
${\bf c}:=(c_1, \cdots, c_m)$, 
$\Psi_t^{(n)}({\bf c}):=
(\Psi_t^{(n)}(c_1), \cdots, \Psi_t^{(n)}(c_m))$, 
and
$\Psi_t({\bf c}):=
(\Psi_t(c_1), \cdots, \Psi_t(c_m))$.
It suffices to show 
(\ref{joint-statement})
with
$(\theta_{nt}(\kappa))_{2 \pi {\bf Z}}$
being replaced by 
$(\tilde{\theta}_{nt}(\kappa))_{2 \pi {\bf Z}}$, 
since 
$(\tilde{\theta}_{nt}(\kappa))_{2 \pi {\bf Z}}$
converges to the uniform distribution by 
Lemma \ref{uniform}. 
By 
Lemmas \ref{tightness}, \ref{tightness-eta}, 
for any fixed
$t_0 > 0$, 
the process 
$\{ (\Psi^{(n)}_t({\bf c}), \eta_t^{(n)}) \}_{n \ge 1}$
on 
$[t_0, \infty)$
is a 
tight family.
Hence we can assume
$(\Psi^{(n)}_t({\bf c}), \eta_t^{(n)})_{t > 0}
\stackrel{d}{\to}
(\Psi_t({\bf c}), \eta_t)_{t > 0}$.
By 
Lemma \ref{independence of eta}
$\eta_{1/n}$
and 
$\eta_t / \eta_{1/n}$
are independent.

We next consider a process
$\tilde{\Psi}_t^{(n)}({\bf c})$
which is defined on 
$[\frac 1n, \infty)$
and is the solution to (\ref{SDEbeta})
with initial value 
$\tilde{\Psi}_{\frac 1n}^{(n)}({\bf c}) = \frac {{\bf c}}{n}$.
\cite{KS} Proposition 4.5
proves the following fact
\[
\sup_{n^{-1} < t < n}|\tilde{\Psi}_t^{(n)}({\bf c}) - \Psi_t({\bf c})| \stackrel{P}{\to} 0, 
\quad
n \to \infty.
\]
Since 
$\eta_{\frac 1n}$
and 
$(\eta_t / \eta_{1/n}, \tilde{\Psi}_t^{(n)}({\bf c}))$
are independent, by letting 
$n \to \infty$, 
it follows that 
$\eta_0 := \lim_{t\downarrow 0}\eta_t$
and 
$(\eta_t / \eta_0, \Psi_t({\bf c}))$
are independent.
Since 
$\eta_0$
is uniformly distributed on 
${\bf U}$,  
$\tilde{\phi}_t := \arg\eta_t 
= \arg\left(\eta_0 \cdot 
\frac {\eta_t}{\eta_0}\right)$
and 
$\Psi_t$
are independent.
\QED
\end{proof}
\subsection{Convergence of 
$\Psi_t^{(n)}$ as increasing functions}
\begin{proposition}
\label{coupling}
Fix any 
$t > 0$. 
Then we can find a coupling such that the following statement is valid for 
a.s.
\[
\lim_{n \to \infty}
(\Psi_t^{(n)})^{-1}(x)=\Psi_t^{-1}(x),  
\quad
\lim_{n \to \infty} (2\theta_{nt}(\kappa))_{2 \pi {\bf Z}}
=
\phi_t
\]
for any 
$x \in {\bf R}$
where 
$\phi_t$
is uniformly distributed and independent of 
$\Psi_t$.
\end{proposition}
As is explained in section 5, 
Proposition \ref{coupling} 
completes the proof of Theorem \ref{sc-limit}. 
To prove
Proposition \ref{coupling}
we shall show below that the convergence 
$\Psi^{(n)}_t \to \Psi_t$
holds in the sense of increasing function-valued process.
Let 
${\cal M}$
be the set of non-negative measures on 
$[a,b]$. 
Fix 
$\{ f_j \}_{j \ge 1}$
a family of smooth functions on 
$[a,b]$
satisfying the property
\[
\mbox{ for }
\omega \in {\cal M} \mbox{ if }
\int_a^b f_j (x) d \omega (x) = 0
\mbox{ for any }
j \ge 1
\Rightarrow
\omega = 0.
\]
We define a metric 
$\rho$
on 
${\cal M}$
by
\[
\rho(\omega_1, \omega_2)
:=
\sum_{j=1}^{\infty}
\frac {1}{2^j}
\left(
\left|
\int_a^b f_j (x) d(\omega_1(x) - \omega_2(x))
\right|
\wedge 1
\right).
\]
Let
\[
\Omega := C([0,T] \to {\cal M})
\]
for 
$T < \infty$. 
We further define 
for a smooth function
$f$
on 
$[a,b]$
a map
$\Phi_f : \Omega \to C([0,T] \to {\bf R})$
by
\begin{eqnarray}
\Phi_f(\omega)(t)
&:=&
\int_a^b f(x) d \omega_t(x)
=
\left[
f(x) \omega_t(x)
\right]_a^b
-
\int_a^b f'(x) \omega_t(x) dx.
\label{Phi-f}
\end{eqnarray}
%
\begin{lemma}
\label{compactness}
Let 
$\{ \mu_n \}_{n \ge 1}$
be a family of probability measures on 
$\Omega$. 
Suppose for each smooth function
$f$
on 
$[a,b]$
a family of probability measures
$\{ \Phi_f^{-1} \mu_n \}_{n \ge 1}$
on 
$C([0,T] \to {\bf R})$
is tight.
Assume further there exists a constant 
$C$
such that
\begin{equation}
{\bf E}_{\mu_n}
\left[
\sup_{0 \le t \le T}
\int_a^b d \omega_t(x)
\right]
\le C
\label{boundedness}
\end{equation}
holds for any
$n \ge 1$. 
Then 
$\{ \mu_n \}_{n \ge 1}$
is tight.
\end{lemma}
\begin{proof}
From
$(\ref{boundedness})$
we see that for any 
$\epsilon > 0$, 
there exists a 
$M > 0$
such that 
\[
\mu_n
\left(
\sup_{0 \le t \le T}
\int_a^b d \omega_t(x) \le M
\right)
\ge 1  - \epsilon.
\]
Set
\[
\Omega_0 :=
\left\{ 
\omega \in \Omega \, \Biggl| \,
\sup_{0 \le t \le T}
\int_a^b d \omega_t (x) \le M
\right\}.
\]
Since 
$\{ \Phi_{f_j}^{-1} \mu_n \}_{n \ge 1}$
is tight for each 
$j \ge 1$, 
there exists a compact set 
${\cal K}_j$ 
in 
$C([0,T] \to {\bf R})$
such that 
$
\mu_n \left(
\Phi_{f_j}^{-1} ({\cal K}_j)
\right)
>
1 - \frac {\epsilon}{2^j}.
$
Set 
\[
{\cal K}:=
\bigcap_{j=1}^{\infty}
\Phi_{f_j}^{-1} ({\cal K}_j) \cap \Omega_0
\subset \Omega.
\]
Then
\begin{equation}
\mu_n({\cal K}^c)
\le
\sum_{j=1}^{\infty}
\mu_n \left(
\Phi_{f_j}^{-1}({\cal K}_j^c
\right)
+
\mu_n (\Omega_0^c)
\le
\epsilon
+
\sum_{j=1}^{\infty}
\frac {\epsilon}{2^j}
=
2 \epsilon.
\label{epsilon}
\end{equation}
We show 
${\cal K}$
is compact in 
$\Omega$. 
Let 
$\{ \omega_n \}_{n \ge 1}$
be a sequence in 
${\cal K}$.
Since
${\cal K}_1$
is compact, there exists a subsequence 
$\{ n_i^1 \}$
along which
$\Phi_{f_1}\left( 
\omega_{n_i^1}
\right)$
is uniformly convergent in 
$C([0,T] \to {\bf R})$. 
Then, 
using the compactness of 
${\cal K}_2$
we can find a subsequence 
$\{ n_i^2 \}$
of 
$\{ n_i^1 \}$
along which 
$\Phi_{f_2} \left( \omega_{n_i^2} \right)$
is uniformly convergent in 
$C([0,T] \to {\bf R})$.
Continuing this procedure for each 
$j$
we find a subsequence 
$\{ n_i^j \}$
of 
$\{ n_i^{j-1} \}$
along which 
$\Phi_{f_j} \left( n_i^j  \right)$
is uniformly convergent in 
$C([0,T] \to {\bf R})$. 
Let 
$m_i = n_i^i$. 
Then for each 
$j\ge 1$, 
$\Phi_{f_j} (\omega_{m_i})$
converges uniformly in 
$C([0,T] \to {\bf R})$. 
Since, 
for any 
$f \in C[a,b]$ 
and 
$\epsilon' > 0$, 
there exists a finite linear combination 
$g$
of 
$\{ f_j \}$
such that 
$
\sup_{x \in [a,b]} 
| f(x) - g(x) | < \epsilon'.
$
We easily have 
\[
\sup_{t \in [0,T]}
| \Phi_f (\omega_{m_i})(t) - \Phi_g (\omega_{m_i})(t) |
\le \epsilon' M
\]
where we have used
$
\int_a^b d \omega_t (x) \le M
$
for any 
$\omega \in {\cal K}$. 
Therefore we see that the limit
$
\lim_{i \to \infty} \Phi_f (\omega_{m_i})(t)
$
exists uniformly with respect to 
$t\in [0,T]$, 
which implies that there exists a 
$\omega \in \Omega$
satisfying 
\[
\int_a^b d \omega_t (x) \le M
\mbox{ and }
\lim_{i \to \infty}
\Phi_f (\omega_{m_i}) 
=
\Phi_f (\omega) (t)
\]
for any 
$t \in [0,T]$
and 
$f \in C([a,b])$. 
Consequently 
we have the compactness of 
${\cal K}$
which together with 
(\ref{epsilon})
shows the tightness of 
$\{ \mu_n \}_{n \ge 1}$. 
\QED
\end{proof}
We shall check that the conditions for 
Lemma \ref{compactness}
are satisfied for 
$\Psi^{(n)}_t( \cdot )$. 
The inequality 
(\ref{boundedness})
follows from 
Lemma \ref{a priori2}.
In view of 
(\ref{Phi-f}), 
the required tightness is implied by the following lemma. 
\begin{lemma}
For 
$f \in C^{\infty}(a,b)$
let
\[
g_n(t) := \int_a^b f(x) \Psi_t^{(n)}(x) dx.
\]
Then, as a family of probability measures on 
$C([0,T] \to {\bf R})$, 
$\{ g_n \}_{n \ge 1}$
is tight.
\end{lemma}
\begin{proof}
It is sufficient to show that following two equations.

(1)
$\lim_{A \to \infty}
\sup_n
{\bf P}\left(
| g_n(0) | \ge A 
\right) = 0$, 

(2)
For any
$\rho > 0$, 
\beq
\lim_{\delta \downarrow 0}
\limsup_{n \to \infty}
{\bf P}
\left(
\sup_{|t-s| < \delta}
| g_n(t) - g_n(s) | > \rho
\right)
= 0. 
\eeq
(1) follows 
from Lemma \ref{a priori2}.
By bounding 
$f$, 
the following equation implies 
(2).
\beq
\lim_{\delta \downarrow 0}
\limsup_{n \to \infty}
{\bf P}
\left[
\sup_{|t-s| < \delta}
\int_a^b
\left| 
\Psi_t^{(n)}(x) - \Psi_s^{(n)}(x) 
\right|dx > \rho
\right]
= 0.
\eeq
Here we borrow an argument in 
\cite{KS} Proposition 2.5 : 
We divide 
$[a,b]$
into 
$N$-intervals
$x_j = a + \frac {b-a}{N}x_j$, 
$j = 0, 1, \cdots, N-1$, 
and have
\begin{equation}
\int_a^b
| \Psi_t^{(n)}(x) - \Psi_s^{(n)}(x) |dx
=
\sum_{j=0}^{N-1}
\int_{x_j}^{x_{j+1}}
| \Psi_t^{(n)}(x) - \Psi_s^{(n)}(x) |dx.
\label{division}
\end{equation}
Since 
$\Psi^{(n)}_t(x)$
is increasing with respect to  
$x$, for 
$x \in [x_j, x_{j+1}]$
the integrand is bounded from above by 
\beq
&&
| \Psi_t^{(n)}(x)
-
\Psi_s^{(n)}(x)|
\\
& \le &
\Psi_t^{(n)}(x_{j+1})-\Psi_t^{(n)}(x_j)
+
| \Psi_t^{(n)}(x_{j})-\Psi_s^{(n)}(x_j) |
+
\Psi_s^{(n)}(x_{j+1})-\Psi_s^{(n)}(x_{j}).
\eeq
Substituting it into 
(\ref{division})
yields
\beq
J&:=&
\int_a^b
| \Psi_t^{(n)}(x) - \Psi_s^{(n)}(x) |dx
\\
& \le &
\frac 1N(b-a)
\left(
\Psi_t^{(n)}(b)-\Psi_t^{(n)}(a)
\right)
 + 
(t \leftrightarrow s)
\\
&& \qquad
+
\sum_{j=0}^{N}
\frac 1N(b-a)
| \Psi_t^{(n)}(x_{j})-\Psi_s^{(n)}(x_j) |
=: I + II.
\eeq
Thus we decompose the probability in question into two terms. 
\beq
{\bf P}\left(
\sup_{|t-s| < \delta} J > \rho
\right)
&\le&
{\bf P}\left(
\sup_{|t-s| < \delta} I > \rho/2
\right)
+
{\bf P}\left(
\sup_{|t-s| < \delta} II > \rho/2
\right)
\\
&=:& III + IV.
\eeq
The 
$III$-term
can be estimated by  
Lemma \ref{a priori2}.
\beq
III
& \le &
{\bf P}\left(
\frac {b-a}{N}
\left(
\Psi_t^{(n)}(b) - \Psi_t^{(n)}(a)
\right)
> \rho/4
\right)
+
(t \leftrightarrow s)
\\
& \le &
\frac {4}{\rho}
\cdot
\frac {b-a}{N}
{\bf E}\left[
\left(
\Psi_t^{(n)}(b) - \Psi_t^{(n)}(a)
\right)
\right]
+(t \leftrightarrow s)
\\
& \le &
2 \cdot
\frac {4}{\rho}
\cdot
\frac {b-a}{N}
{\bf E}\left[
\sup_{0 \le t \le T}
\left(
\Psi_t^{(n)}(b) 
\right)
\right]
\le 
\frac CN.
\eeq
Thus for any 
$\epsilon > 0$
we take 
$N$ 
large enough independently of   
$\delta$
to have
$
III < \frac {\epsilon}{2}.
$
For such fixed 
$N$, 
we have
\beq
IV
& \le &
\sum_{j=0}^N
{\bf P}\left(
\frac {b-a}{N}
\sup_{|t-s|<\delta}
| \Psi_t^{(n)}(x_j) - \Psi_s^{(n)}(x_j) |
> \frac {\rho}{2N}
\right)
\\
&=&
\sum_{j=0}^N
{\bf P}\left(
\sup_{|t-s|<\delta}
| \Psi_t^{(n)}(x_j) - \Psi_s^{(n)}(x_j) |
> \frac {\rho}{2(b-a)}
\right). 
\eeq
Since
$\{ \Psi^{(n)}_t(x_j) \}_{j=0}^{N}$
is tight by 
Lemma \ref{tightness}, we can let 
$IV < \epsilon/2$
by taking 
$n$ 
large and then taking 
$\delta > 0$
small. 
\QED
\end{proof}
We identify an element of 
${\cal M}$
with a non-decreasing and right continuous function 
$\omega$
on
$[a,b]$
satisfying 
$\omega(a)= 0$. 
Then 
$\omega_n$
converges to 
$\omega \in \Omega$
if and only if 
$\omega_n (x) \to \omega(x)$
at any point of continuity of 
$\omega$.\\
\begin{lemma}
Suppose 
$\left\{
\omega_{n}
\right\}_{n\geq1}\subset\mathcal{M}$ 
converges to 
$\omega$ 
of 
$\mathcal{M}$. 
Assume 
$\omega$ 
is continuous. 
Then the
convergence is uniform.
\end{lemma}

\begin{proof}
Assume 
$\left\{
\omega_{n}
\right\}_{n\geq1}$ 
does not converge to 
$\omega$
uniformly. 
Then there exists a sequence 
$n_{1}<n_{2}<\cdots$, 
$\left\{
t_{k}
\right\}_{k\geq1}$ 
and a positive number 
$\epsilon_{0}$ 
such that
\begin{equation}
\left\vert 
\omega_{n_{k}
}\left(  t_{k}\right)  
-
\omega\left(  t_{k}\right)
\right\vert \geq\epsilon_{0} \label{6}
\end{equation}
is valid for any 
$k=1,2,\cdots.$ 
We can assume 
$t_{k}\rightarrow t_{0}
\in\left[a,b\right]$ 
keeping 
$t_{1}<t_{2}<\cdots<t_{0}.$ 
Then
\[
\omega_{n_{k}}\left(  t_{l}\right)  -\omega\left(  t_{k}\right)  \leq
\omega_{n_{k}}\left(  t_{k}\right)  -\omega\left(  t_{k}\right)  \leq
\omega_{n_{k}}\left(  t_{0}\right)  -\omega\left(  t_{k}\right)
\]
for any 
$l<k,$ 
hence letting 
$k\rightarrow\infty,$ 
we have
\beq
\omega\left(  t_{l}\right)  -\omega\left(  t_{0}\right)  
&\leq&
\liminf_{k\rightarrow\infty}\left(  \omega_{n_{k}}\left(  t_{k}\right)
-\omega\left(  t_{k}\right)  \right)  
\\
&\leq&
\limsup_{k\rightarrow\infty}
\left(  \omega_{n_{k}}\left(t_{k}\right)  
-\omega\left(  t_{k}\right)
\right)  
\leq
\omega\left(  t_{0}\right)  -\omega\left(  t_{0}\right)  =0.
\eeq
Consequently, letting 
$l\rightarrow\infty,$ 
we see
$
\lim_{k\rightarrow\infty}\left(  \omega_{n_{k}}\left(  t_{k}\right)
-\omega\left(  t_{k}\right)  \right)  =0,
$
which contradicts (\ref{6}).
\end{proof}

{\it Proof of Proposition \ref{coupling}}\\
By
Lemma \ref{compactness}, 
the sequence of increasing function-valued process
$\{ \Psi_t^{(n)}(\cdot)\}_n$
is tight. 
Hence 
$(\Psi_t^{(n_k)}, (2\theta_{n_k t})_{2 \pi {\bf Z}})
\stackrel{d}{\to}
(\Psi_t, \phi_t)$
for some subsequence 
$\{ n_k \}$. 
By Skorohod's theorem, 
we can suppose 
$(\Psi_t^{(n_k)}, (2\theta_{n_k t})_{2 \pi {\bf Z}}) \stackrel{a.s.}{\to} (\Psi_t, \phi_t)$.
Hence in particular we fix any 
$t > 0$ 
and obtain
\[
\rho(\Psi_t^{(n_k)}, \Psi_t)
=
\sum_{j \ge 1} \frac {1}{2^j}
\left(
\left|
\int_a^b f_j (x) d (\Psi_t^{(n_k)}(x) -  \Psi_t(x))
\right|
\wedge 1
\right)
\stackrel{n \to \infty}{\to} 0, 
\quad
a.s.
\]
By 
Lemma \ref{parameter-continuity}
$\Psi_t$
is continuous and increasing. 
Hence for a.s., 
$\Psi_t^{(n_k)}(x) \to \Psi_t(x)$
holds for any 
$x$. 
Moreover by 
Lemma \ref{inverse}
$(\Psi_t^{(n_k)})^{-1}(x) \stackrel{a.s.}{\to} \Psi_t^{-1}(x)$.
Therefore 
Proposition \ref{coupling}
is proved. 
\QED

\vspace*{1em}
\noindent {\bf Acknowledgement }
This work is partially supported by 
JSPS grants Kiban-C no.22540163(S.K.) and no.22540140(F.N.).

%

\end{document}